\documentclass[journal]{IEEEtran}
\ifCLASSINFOpdf

\else
  
\fi

\usepackage{graphicx}
\usepackage{cite}
\usepackage{amsmath} 
\usepackage{amssymb}
\usepackage{amsthm}
\usepackage{xcolor}
\usepackage{url}
\usepackage{mathrsfs}
\usepackage{algorithm}
\usepackage{algpseudocode}
\usepackage{multirow}
\usepackage{booktabs}

\usepackage{enumitem}

\usepackage{tabularx}

\usepackage{xcolor}

\newcommand{\diag}{\mathop{\mathrm{diag}}}
\DeclareMathOperator{\blkdiag}{blkdiag}

\newcommand\blfootnote[1]{%
\begingroup
\renewcommand\thefootnote{}\footnote{#1}%
\addtocounter{footnote}{-1}%
\endgroup
}

\newtheorem{proposition}{Proposition}
\newtheorem{corollary}{Corollary}

\begin{document}

\title{Modern Base Station Architecture: Enabling Passive Beamforming with Beyond Diagonal RISs}

\author{Mahmoud~Raeisi,~\IEEEmembership{Member,~IEEE,}
        Hui~Chen,~\IEEEmembership{Member,~IEEE,}
        Henk~Wymeersch,~\IEEEmembership{Fellow,~IEEE,}
        ~and~Ertugrul~Basar,~\IEEEmembership{Fellow,~IEEE}
\thanks{Mahmoud Raeisi is with the Department of Electrical and Electronics Engineering, Istanbul Medipol University, 34810 Istanbul, Turkey. (e-mail: mahmoud.raeisi@medipol.edu.tr)}
\thanks{Hui Chen and Henk Wymeersch are with the Department of Electrical Engineering, Chalmers University of Technology, 41296 Gothenburg, Sweden (e-mail: hui.chen@chalmers.se; henkw@chalmers.se).}
\thanks{E. Basar is with the Department of Electrical Engineering, Tampere University, 33720 Tampere, Finland, on leave from the Department of Electrical and Electronics Engineering, Koc University, 34450 Sariyer, Istanbul, Turkey (email: ertugrul.basar@tuni.fi).}

\thanks{The work of Hui Chen and Henk Wymeersch is supported by the Swedish Research Council under VR grant 2022-03007 and the SNS JU project 6G-DISAC under the EU’s Horizon Europe Research and Innovation Program under Grant Agreement No 101139130.}}


\maketitle

\begin{abstract}
Beamforming plays a crucial role in millimeter wave (mmWave) communication systems to mitigate the severe attenuation inherent to this spectrum. However, the use of large active antenna arrays in conventional architectures often results in high implementation costs and excessive power consumption, limiting their practicality. As an alternative, deploying large arrays at transceivers using passive devices, such as reconfigurable intelligent surfaces (RISs), offers a more cost-effective and energy-efficient solution.
In this paper, we investigate a promising base station (BS) architecture that integrates a beyond diagonal RIS (BD-RIS) within the BS to enable passive beamforming. 
By utilizing Takagi's decomposition and leveraging the effective beamforming vector, the RIS profile can be designed to enable passive beamforming directed toward the target.
Through the beamforming analysis, we reveal that BD-RIS provides robust beamforming performance across various system configurations, whereas the traditional diagonal RIS (D-RIS) exhibits instability with increasing RIS size and decreasing BS-RIS separation—two critical factors in optimizing RIS-assisted systems. Comprehensive computer simulation results across various aspects validate the superiority of the proposed BS-integrated BD-RIS over conventional D-RIS architectures, showcasing performance comparable to active analog beamforming antenna arrays.
\end{abstract}

\begin{IEEEkeywords}
Beyond diagonal RIS, millimeter wave, bit error rate, achievable rate, passive beamforming.
\end{IEEEkeywords}
\IEEEpeerreviewmaketitle

\section{Introduction}\label{Sec:Introduction}

\IEEEPARstart{T}{he} \textcolor{black}{growing demand for high-capacity and low-latency communication has accelerated the development of advanced wireless systems. To support emerging applications, higher frequency bands such as the millimeter-wave (mmWave) spectrum have gained traction for sixth-generation (6G) networks~\cite{10176315, wang2020joint, raeisi2022cluster}. The wide bandwidth available at mmWave frequencies enables high data rates, while the short wavelength supports large antenna arrays capable of forming highly directional beams, which are essential for mitigating severe path loss and facilitating spatial multiplexing~\cite{mahmood20223, 8891298}.}

Despite these advantages, mmWave systems often require costly and power-hungry hardware, especially when using fully digital beamforming. 
\textcolor{black}{Hybrid beamforming has been widely recognized as a cost-effective alternative to fully digital beamforming in mmWave systems. Numerous studies have sought to reduce hardware costs and system complexity by exploring scenarios such as angular-based precoding~\cite{koc2020hybrid}, three-dimensional (3D) antenna array designs~\cite{mahmood20223}, and user grouping strategies~\cite{8891298}. To further minimize hardware requirements, fixed phase shifter architectures have been introduced~\cite{yu2018hardware}. However, massive MIMO deployments still demand a large number of phase shifters, resulting in increased cost and power consumption.}

On the other hand, reconfigurable intelligent surfaces (RISs) have emerged as a low-cost alternative that can dynamically shape the wireless environment~\cite{basar2019wireless, raeisi2024comprehensive}. While RISs are traditionally viewed as auxiliary elements, recent studies have explored their potential as primary beamformers at the base station (BS) or user equipment (UE) \cite{10515204}. This approach opens the door to RIS-integrated transceivers that reduce circuit complexity and power consumption, enabling fully passive analog beamforming.



\subsection{Related Works}


Programmable metasurfaces have gained attention for their passive nature and ability to enable low-cost arrays that direct incident electromagnetic (EM) waves. Explored types include passive RIS \cite{basar2019wireless}, active RIS \cite{9998527}, simultaneously transmitting and receiving (STAR) RIS \cite{9570143}, stacked intelligent metasurface (SIM), and beyond diagonal (BD) RIS \cite{li2023reconfigurable}. 
Some studies propose RISs as primary beamformers at the BS or UE to simplify massive MIMO deployment \cite{10515204, raeisi2024comprehensive}. Furthermore, integrating RIS directly within transceivers mitigates the multiplicative path loss challenges typically associated with conventional RIS-assisted systems \cite{raeisi2024comprehensive}.

\textcolor{black}{User-side RIS integration has been explored as a means to realize energy- and cost-efficient large-scale arrays~\cite{9685418, 9598898}. In~\cite{10144102}, the authors propose embedding an RIS within the BS radome as an auxiliary passive array, enabling real-time reconfiguration to enhance performance cost-effectively. In~\cite{9991837}, a hybrid beamforming approach is developed to suppress sidelobes and optimize beam patterns using a least-squares-based method, jointly adjusting BS beamforming and transmissive RIS phases. A key advancement in this domain is SIM, which seamlessly integrates RISs within BS or UE hardware to passively support advanced MIMO functionalities. By leveraging multilayer meta-atom structures, SIM enables scalable passive arrays at transceivers using low-cost components~\cite{10515204, 10535263, 10279173, 10679332, 10445164}.}


\textcolor{black}{BD-RIS represents an advanced evolution of RIS technology, enabling control over both the phase and amplitude of impinging signals without requiring active amplification or attenuation~\cite{li2023reconfigurable, 9913356, 10158988}. Unlike diagonal RIS (D-RIS), BD-RIS incorporates inter-element connectivity, allowing amplitude adjustments across elements and enhancing beamforming capabilities. This makes BD-RIS a compelling candidate for transceiver integration, where its ability to perform passive analog beamforming can replace conventional phase shifters, thereby simplifying transceiver design while maintaining high performance. In addition, other non-diagonal RIS architectures have been investigated, including non-diagonal scattering-based RIS designs \cite{9737373}, group-connected designs \cite{10472097}, and joint-sector-based configurations \cite{10643263}, aiming to enhance performance while reducing hardware complexity. 
As shown in~\cite{9514409}, BD-RIS can outperform D-RIS under rich-scattering conditions where the BS-RIS and RIS-UE elements exhibit linear independence. However, this performance advantage diminishes in sparsely scattered environments.}


\textcolor{black}{Initial studies have explored the potential of BD-RIS integration at the BS. In~\cite{10308579}, a reflective BD-RIS in the sub-6 GHz band is shown to effectively mitigate inter-user interference in multi-user systems. The work in~\cite{10693959} examines BD-RIS for integrated sensing and communication (ISAC), but models the mmWave channel using rich scattering, which does not fully capture the sparse nature of mmWave propagation. Nevertheless, the adoption of a multi-user setup enriches the channel and enhances BD-RIS performance over D-RIS. In~\cite{10530995}, a single-layer SIM with BD-RIS outperforms all SIM designs based on D-RIS, highlighting BD-RIS’s potential in RIS-integrated systems. However, the RIS-UE channel is again modeled using Rayleigh fading. Finally, \cite{raeisi2024efficient} demonstrates that integrating a linear BD-RIS at the BS enhances beamforming gain compared to D-RIS, though the study is limited to linear arrays and localization estimation.}


Integrating BD-RIS within the BS seamlessly incorporates it as part of the BS, rendering it indistinguishable as a separate unit from the perspective of the UE. Consequently, a fundamental step in designing such a structure involves evaluating system performance under varying BS design parameters, such as array size and BS-RIS separation. In this context, the BS-RIS channel assumes a pivotal role in the design process, as it is directly influenced by these parameters and remains entirely under the service provider’s control. While initial studies \cite{10308579}, \cite{10693959}, \cite{10530995} have explored the concept of BS-integrated BD-RIS, they fall short of analyzing the impact of these design parameters on the proposed architecture. Notably, the primary advantage of BD-RIS over D-RIS lies in its ability to adjust the amplitude of impinging signals \cite{li2023reconfigurable}. Given the close proximity of RIS to the active antenna at the BS, the amplitude variations of signals emitted by the active antenna across different RIS elements become significant. Thus, it is reasonable to anticipate that BD-RIS would outperform D-RIS without factoring \textcolor{black}{environmental scattering conditions in RIS-UE channel.} Furthermore, employing a Rayleigh channel model for the RIS-UE link or extending the analysis to multi-user scenarios, as partially considered in \cite{10308579}, \cite{10693959}, \cite{10530995}, introduces additional amplitude variations across the RIS elements, further amplifying the performance gains of BD-RIS over D-RIS.

\subsection{Motivations and Contributions}\label{Sec:GapsMotivations}


The primary motivation for deploying large arrays and massive MIMO in high-frequency communication systems is forming narrow beams toward intended targets, effectively mitigating severe path loss and high attenuation. 
The influence of BS-integrated BD-RIS on beamforming gain, however, remains unexplored in the aforementioned studies. Moreover, the mmWave communication environment is characterized by spatially sparse scattering, making the assumption of a rich-scattering Rayleigh channel model unrealistic and disconnected from practical scenarios.
With this in mind, the main contributions of this paper are summarized as follows:

\begin{table*}[h]
\color{black}
\scriptsize
\centering
\caption{Comparative summary of existing BS-integrated BD-RIS works and the proposed system, highlighting key features in terms of architecture, scattering conditions, performance metrics, and system-level insights.}
\begin{tabularx}{\textwidth}{|c|c|c|c|c|c|X|} 
    \hline
    \textbf{Ref.} & \textbf{Architecture} & \textbf{System} & \textbf{Channel}  & \textbf{Scattering} & \textbf{Metric} & \textbf{Highlights} \\
    \hline \hline
    \cite{10308579} & Reflective & Multi User & Pure LOS  & Rich & Spectral Efficiency & Inter-user interference suppression \\
    \hline
    \cite{10693959} & Transmissive & Multi User & Rician  & Rich & Sum Rate & Joint sensing and communication gain  \\
    \hline
    \cite{10530995} & Transmissive & Single User & Rayleigh  & Rich & Normalized Channel Gain & SIM gain validation for BD-RIS \\
    \hline
    \cite{raeisi2024efficient} & Transmissive & Single User & Pure LOS & Sparse & CRB  & Localization performance  \\
    \hline 
    \textit{This paper} & \textit{Transmissive} & \textit{Single User} & \textit{Geometric} & \textit{Sparse} & \textit{BER, Achievable Rate} & \textit{Comprehensive evaluation, robust beamforming, and geometry-guided grouping strategy} \\
    \hline
\end{tabularx}
\label{tab:Contributions}
\end{table*}

\begin{itemize}
    \item \textit{BS-Enabled Passive Beamforming}. We investigate a BS-integrated BD-RIS system model that facilitates passive beamforming at the BS by emulating a multiple-input single-output (MISO) configuration through a single-input single-output (SISO)-assisted BD-RIS. 
    \textcolor{black}{Utilizing the widely adopted geometric cluster-based channel model, which effectively captures the spatially sparse scattering typical of mmWave environments, we demonstrate that BD-RIS can form narrow beams with beamforming gains comparable to those of traditional active antenna arrays. In contrast, D-RIS performance is limited, particularly when a large number of elements are deployed in shorter distances to the active BS antennas. This critical aspect has not been explored in \cite{10308579, 10693959, 10530995}.}  

    \item \textit{Geometric Grouping Strategy}. \textcolor{black}{We demonstrate that the internal geometry of the BS-integrated BD-RIS plays a crucial role in reducing the circuit complexity of the proposed system. Importantly, achieving beamforming gains close to those of active antenna arrays does not require a fully-connected BD-RIS architecture. Instead, by organizing BD-RIS elements into groups that are symmetric with respect to the active antenna, high beamforming performance can be maintained. This is because the channel amplitude variations of the BS-RIS link within each group remain consistent with those observed across the entire BD-RIS. 
    To the best of our knowledge, such a geometry-aware, group-connected BD-RIS design in a BS-integrated configuration has not been previously investigated.}

    \item \textit{\textcolor{black}{Mathematical Analysis and} Simulation Insights}. We evaluate the beamforming gain of the proposed BS-integrated BD-RIS configured as a uniform planar array, benchmarking its performance against traditional D-RIS and active analog beamforming antenna arrays. Through the introduction of a novel metric, \textit{channel amplitude variations}, we highlight and address the fundamental limitation of BS-integrated D-RIS systems. 
    \textcolor{black}{We provide a formal mathematical proof indicating that integrating BD-RIS instead of D-RIS inside the BS guarantees a minimum performance gain, which is independent of the scattering conditions and is fully determined by the BS design parameters, i.e., array size and BS-RIS separation. We further demonstrate that a substantial part of the gain originates from the BS design parameters rather than the environmental scattering conditions.} To the best of our knowledge, this aspect was unknown in the previous works.
    Our findings also reveal that the proposed BS-integrated BD-RIS achieves a beamforming gain and beam pattern comparable to the performance of active analog beamforming antenna arrays. In contrast, D-RIS exhibits limited beamforming capability and wider beam patterns. Besides, we conducted various simulations to evaluate the average bit error rate (ABER) and validate the results with the theoretical upper bound, along with a comprehensive achievable rate analysis across various parameters under various conditions.
\end{itemize}
\textcolor{black}{To highlight our contributions in comparison to related works in the literature, Table \ref{tab:Contributions} provides a detailed summary of the considered BS-integrated BD-RIS-based architectures, key system models, channel characteristics, performance metrics, and unique insights offered by each study.}

The paper is organized as follows: Section \ref{sec:SystemChannelSignalModel} introduces the system, channel, and signal models. Section \ref{sec:BDRIS_Circuit_Configuration} outlines the BD-RIS architecture and its configuration algorithm for fully-connected and group-connected modes. Section \ref{sec:CAV} highlights key differences between BD-RIS and conventional D-RIS, discussing the impact of design parameters on performance. Section \ref{Sec:SimulationResutls} presents simulation results, comparing the proposed BS-integrated BD-RIS with benchmarks. Finally, Section \ref{Sec:Conclusion} concludes the paper.


\blfootnote{
\textit{Notation}: Bold uppercase and lowercase letters are used to denote matrices and vectors, respectively. The operators $(\cdot)^\mathsf{T}$, $(\cdot)^{*}$, $(\cdot)^\mathsf{H}$, $\|\cdot\|$, $\vert\cdot\vert$, $\angle$, $\diag(\cdot)$, and $\blkdiag(\cdot)$ denote transpose, conjugate, Hermitian, norm, absolute value, phase of a complex number, diagonalization, and block-diagonalization, respectively. $\mathbb{C}$ represents the set of complex numbers while $\odot$ denotes the Hadamard (element-wise) product. The operators $\mathbb{E}_X[\cdot]$ and $\mathbb{V}_X[\cdot]$ indicate the expectation and variance taken over the random variable $X$, while $\mathbb{P}(\cdot)$ represents the probability of an event. The operators $\cup$ and $\cap$ denote set union and intersection, respectively, while $\emptyset$ denotes the empty set. The notation $[\cdot]_i$ refers to the $i$-th element of a vector, while $\boldsymbol{v}_{[i:j]}$ defines a sub-vector containing elements from position $i$ to $j$. Furthermore, $[a:\Delta:b]$ denotes a discrete sequence ranging from $a$ to $b$ with increments of $\Delta$. $\mathcal{CN}(\mu,\sigma^2)$ denotes a complex Gaussian distribution with mean $\mu$ and variance $\sigma^2$, $\mathcal{L}(\varpi, \varsigma)$ refers to a Laplace distribution with location parameter $\varpi$ and scale parameter $\varsigma$, and $\mathcal{U}(a,b)$ denotes a uniform distribution within the interval $[a, b]$. The $Q$-function is defined as $Q(x) = \frac{1}{\sqrt{2 \pi}}\int_x^{\infty} \exp{(-\frac{u^2}{2})} du$. $\boldsymbol{I}_n$ denotes the $n \times n$ identity matrix, and $\jmath = \sqrt{-1}$ represents the imaginary unit.
}

\section{System, Channel, and Signal Model}\label{sec:SystemChannelSignalModel}
In this section, we begin by presenting the system model, outlining the architecture of the proposed BS-integrated BD-RIS, the arrangement of its passive elements, and its role in enabling passive beamforming for efficient directional communication in the mmWave environment. Next, we detail the channel model, capturing the interaction between the active antenna and the BD-RIS, as well as the propagation characteristics of the wireless channel between the BD-RIS and the UE. We then describe the signal model, explaining the transmission of symbols, passive beamforming by the BD-RIS, and the formulation and detection of the received signal at the UE.

\subsection{System Model}\label{sec:system_model}

Fig. \ref{fig:SystemModel} illustrates the proposed system model and the BS structure designed for deployment in a mmWave environment. 
Building upon the motivations discussed in Section \ref{Sec:Introduction}, we consider a point-to-point (P2P) SISO system where the BS is equipped with a single active antenna and a passive BD-RIS integrated into the BS as depicted in Fig. \ref{fig:SystemModel}(a). This configuration enables passive beamforming, effectively forming the beams to combat the severe path loss and attenuation characteristic of the mmWave spectrum \cite{10176315, 10188340, raeisi2024comprehensive}. Fig. \ref{fig:SystemModel}(b) provides a 3D view of the BD-RIS side facing the active antenna, with the opposite side, which faces outward, having an identical structure. Each side of the BD-RIS is referred to as a sector, as illustrated in Fig. \ref{fig:SystemModel}(c).
For the proposed application, the BD-RIS operates exclusively in the transmissive mode (signals are sending to sector 2), with its reflective mode deactivated \cite{9913356}.{\color{black}\footnote{\color{black}The feasibility of implementing a transmission-mode BD-RIS without a dedicated mode-switching bias line has been demonstrated in \cite{ming2025hybrid}, using a design based on 2-bit phase-reconfigurable RF antennas realized through RF switches. While these switches still require a small DC bias, the overall design supports a hardware-efficient implementation consistent with the system assumed in this work.}}

\begin{figure}
    \centering
    \includegraphics[width=\columnwidth]{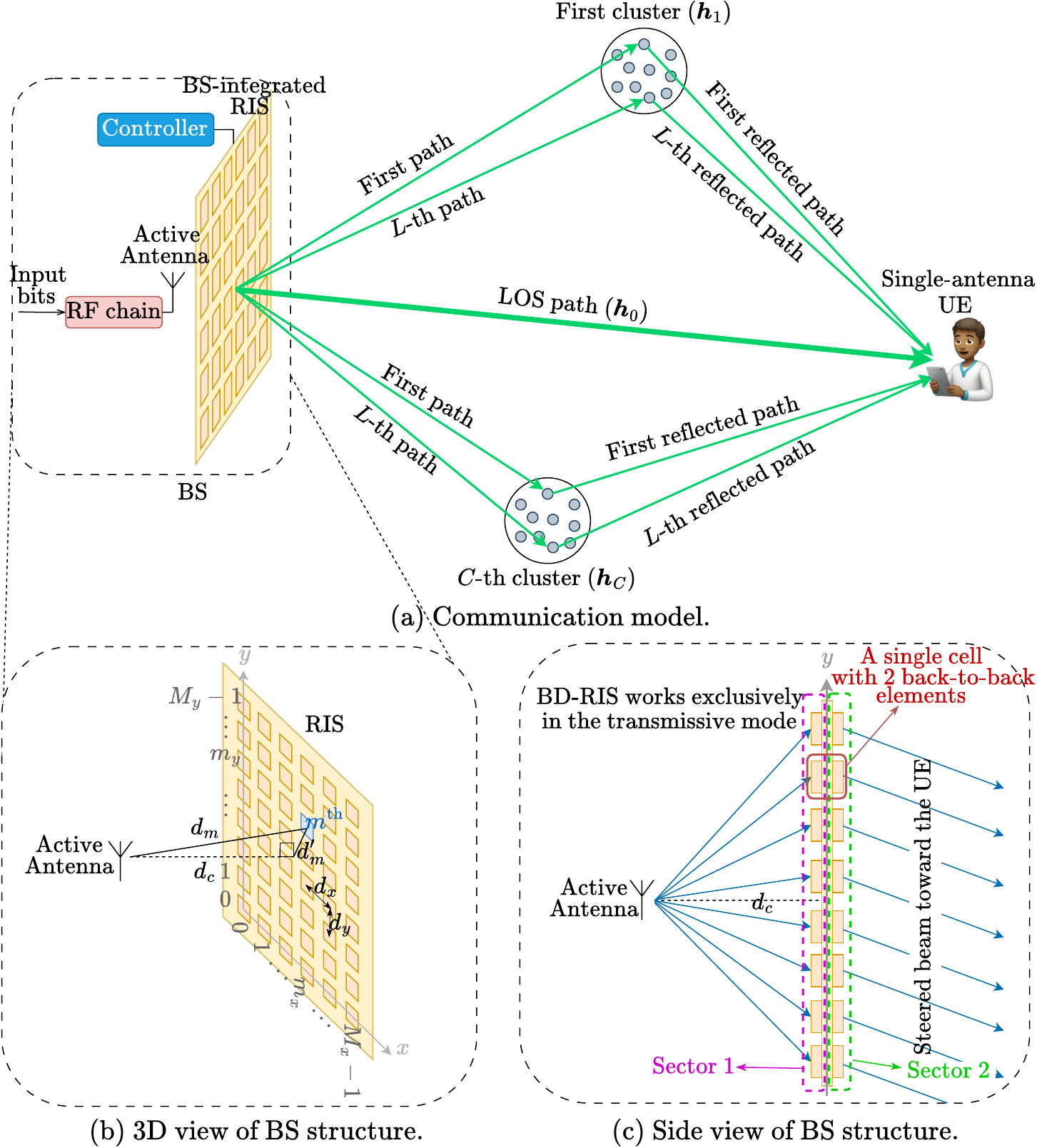}
    \caption{System model: (a) Communication model showcasing the BS-integrated BD-RIS and single-antenna UE, illustrating the LOS path and NLOS clusters with multipath components; (b) 3D view of the BS structure, highlighting the BD-RIS geometry, BS-RIS separation ($d_c$), and the active antenna-to-RIS element distances ($d_m$); (c) Side view of the BS structure, demonstrating the BD-RIS operation in transmissive mode, where Sector 1 receives the incident wave, and Sector 2 transmits the steered beam toward the UE.}
    \label{fig:SystemModel}
\end{figure}

To facilitate exclusive transmissive operation, the passive elements on each side are organized into a uniform planar array with  $M_x$  columns and  $M_y$  rows. Each back-to-back pair of passive elements from the two sectors of the BD-RIS constitutes a single cell, enabling seamless coordination between the two sectors.
This arrangement comprises a total of  $M = M_x \times M_y$  cells, allowing the integrated BD-RIS to effectively steer the EM wave radiated by the active antenna toward the UE.
To maintain focus on the proposed system model, the details of the BD-RIS circuit are provided separately in Section \ref{sec:GroupConnected}.

The spacing between adjacent columns along the $x$-axis is $d_x$, while the spacing between adjacent rows along the $y$-axis is $d_y$. The $m$-th cell ($0 \leq m \leq M-1$) is located at the intersection of the $m_y$-th row ($0 \leq m_y \leq M_y - 1$) and $m_x$-th column ($0 \leq m_x \leq M_x - 1$), where the index $m$ is given by $m = m_y M_x + m_x$. 
The active antenna is positioned along an axis perpendicular to the RIS plane, passing through the geometric center of the RIS surface at a distance of $d_c$, as shown in Fig. \ref{fig:SystemModel}(b). Additionally, the distance between the active antenna and the $m$-th cell is denoted as $d_m$ and defined as follows:
\begin{equation}
    d_m = \sqrt{d_c^2 + d_m'^2},
\end{equation}
where $d_m'$ represents the distance between the  $m$-th element and the center of the BD-RIS, defined as:
\begin{equation}
    d_m' = \frac{1}{2}\sqrt{(d_x|2m_x - M_x + 1|)^2 + (d_y|2m_y - M_y + 1|)^2}.
\end{equation}

\subsection{Channel Model} 
In the proposed system, the channel between the active antenna and the BD-RIS plays a critical role in shaping the performance of passive beamforming. The integration of the BD-RIS within the BS introduces unique propagation characteristics, as the EM wave emitted by the active antenna interacts with each BD-RIS cell, arriving with distinct amplitudes and phases. To capture this interaction, the channel coefficient vector $\boldsymbol{g} \in \mathbb{C}^{M \times 1}$ from the active antenna to the BD-RIS is derived using the Rayleigh-Sommerfeld diffraction theory for near-field propagation \cite{10158690, 10643881, 10557708, 10622385, 10445164} as follows:
\begin{equation}\label{eq:Rayleigh-Sommerfeld}
    \left[\boldsymbol{g}\right]_m = \frac{A d_c}{d_m^2} \left( \frac{1}{2 \pi d_m} -  \frac{\jmath}{\lambda} \right) e^{  \jmath 2 \pi d_m / \lambda},
\end{equation}
where $A$ represents the area of each passive element, and  $\lambda$ represents the wavelength of the carrier signal.{\color{black}\footnote{\color{black}In this work, we adopt the unilateral approximation, assuming negligible electrical interaction between the BD-RIS and the active antenna, following the modeling framework in \cite{10530995}. This choice is further supported by the typical antenna–RIS spacing ($\geq \lambda/2$), which, as in conventional array architectures, ensures weak mutual coupling.}}

After traversing the BD-RIS, the signal propagates through the wireless channel between the BD-RIS and the UE. In the mmWave environment, such channels are characterized by significant propagation challenges, including penetration loss, severe path loss, and attenuation, which create a sparsely scattered propagation environment. To capture these conditions accurately, the widely recognized clustered geometric channel model is employed. This model accounts for the combined effect of a  line-of-sight (LOS) path and multiple non-LOS (NLOS) clusters, as illustrated in Fig. \ref{fig:SystemModel}(a). Specifically, the channel between the BD-RIS and UE, denoted as $\boldsymbol{h} \in \mathbb{C}^{M \times 1}$, is modeled as a superposition of these paths and is expressed as follows \cite{10176315, 10188340, akdeniz2014millimeter, ying2020gmd}:
\begin{equation}\label{eq:ChannelModel}
    \begin{aligned}
        \boldsymbol{h} & = \sqrt{\frac{M}{C L + 1}} \Bigg( 
            \underbrace{\alpha \boldsymbol{a}^\mathsf{H} (\varphi_0,\vartheta_0)}_{\boldsymbol{h}_0^{\mathsf{T}}} 
            + \sum_{c=1}^{C} 
            \underbrace{\sum_{\ell = 1}^{L} \beta_{c,\ell} \boldsymbol{a}\mathsf{^H}(\varphi_{c,\ell},\vartheta_{c,\ell})}_{\boldsymbol{h}_c^{\mathsf{T}}} 
        \Bigg)^{\mathsf{T}} \\
        &= \sqrt{\frac{M}{C L + 1}} \Bigg( \underbrace{\boldsymbol{h}_0}_{\textrm{LOS component}} + \underbrace{\sum_{c = 1}^C \boldsymbol{h}_c}_{\textrm{NLOS component}} \Bigg),
    \end{aligned}
\end{equation}
where $C$ represents the number of NLOS clusters and $L$ denotes the number of NLOS paths in each cluster. 
The parameters $\alpha \sim \mathcal{CN}(0, \sigma_{\alpha}^2)$, $\varphi_0 \sim \mathcal{U}[-\pi, \pi]$, and $\vartheta_0 \sim \mathcal{U}[0, \frac{\pi}{2}]$ denote the complex path gain, azimuth angle of departure (AOD), and elevation AOD for the LOS path, respectively  \cite{10176315, 10188340}. For the $\ell$-th path in the $c$-th cluster, the corresponding parameters are $\beta_{c,\ell} \sim \mathcal{CN}(0, \sigma_{\beta}^2)$, $\varphi_{c,\ell} \sim \mathcal{L}(\varphi_c, \varsigma_c)$, and $\vartheta_{c,\ell} \sim \mathcal{L}(\vartheta_c, \varsigma_c)$, representing the complex path gain, azimuth AOD, and elevation AOD, respectively \cite{10176315, 10188340, el2014spatially}.
Here, $\varphi_c \sim \mathcal{U}[-\pi, \pi]$ and $\vartheta_c \sim \mathcal{U}[0, \frac{\pi}{2}]$ are the mean azimuth and elevation AODs while $\varsigma_c$ denotes the angular spread in both the azimuth and elevation dimensions for the $c$-th cluster \cite{10176315, 10188340, el2014spatially}. The variances of the LOS and NLOS path gains are defined by $\sigma_i^2 = 10^{-0.1 \textrm{PL}_i(d)}$, $i \in \{ \alpha, \beta \}$, where $\textrm{PL}_i(d)$ represents the path loss between the BD-RIS and UE, and $d$ denotes the distance between them.
In addition, $\boldsymbol{a}(\phi, \theta) \in \mathbb{C}^{M \times 1}$ represents the array response (steering) vector at the BD-RIS transmit terminal, indicating that the corresponding path is directed toward the azimuth angle $\phi$ and elevation angle $\theta$. The mathematical expression for $\boldsymbol{a}(\phi, \theta)$ is given as follows:
\begin{equation}\label{eq:steering_vector}
    \boldsymbol{a}(\phi,\theta) = \frac{1}{\sqrt{M}} [e^{\jmath \boldsymbol{k}\mathsf{^T} \boldsymbol{p}_0}, e^{\jmath \boldsymbol{k}\mathsf{^T} \boldsymbol{p}_1}, \dots, e^{\jmath \boldsymbol{k}\mathsf{^T} \boldsymbol{p}_{M-1}}]\mathsf{^T},
\end{equation}
where $\boldsymbol{k}$ represents the wave-number vector, which characterizes the direction of the corresponding path, and $\boldsymbol{p}_m$, $0 < m < M-1$, denotes the position vector of the $m$-th cell, specifying its location on the BD-RIS. These vectors are mathematically expressed as:
\begin{equation}
    \boldsymbol{k} = \frac{2\pi}{\lambda}[\sin{\theta} \cos{\phi}, \sin{\theta} \sin{\phi}]\mathsf{^T},
\end{equation}
\begin{equation}
    \boldsymbol{p}_m = [m_x d_x, m_y d_y]\mathsf{^T}.
\end{equation}
Notably, in (\ref{eq:ChannelModel}), the channel model is decomposed into the LOS sub-channel, represented by $\boldsymbol{h_0}$, and the summation of NLOS sub-channels, $\sum_{c = 1}^C \boldsymbol{h}_c$.

\subsection{Signal Model}
This paper focuses on P2P communication, where a single stream of information bits is transmitted to a single-antenna UE. The transmitted symbol $s$ is chosen from an $\mathcal{M}$-ary constellation set of order $\mathcal{M}$ and transmits to the BD-RIS through the transmission channel $\boldsymbol{g}$. As a result, the system’s spectral efficiency is defined as $\eta = \log_2 \mathcal{M}$ bits per channel use (bpcu). The BD-RIS applies a non-diagonal configuration scattering matrix $\boldsymbol{\Omega} \in \mathbb{C}^{M \times M}$ to the incoming EM wave from the active antenna, effectively steering it toward the UE through the mmWave channel $\boldsymbol{h}$. The BD-RIS configuration strategy is detailed in Section \ref{Sec: BD-RIS Config}.
The received signal at the UE can be expressed as:
\begin{equation}\label{eq:signal model}
    y = \sqrt{P} \boldsymbol{h}^{\mathsf{T}} \boldsymbol{\Omega} \boldsymbol{g} s + n = \sqrt{P} \boldsymbol{h}^{\mathsf{T}} \boldsymbol{\zeta} s + n,
\end{equation}
where $P$ is the BS transmitted power, and $n \in \mathcal{CN}(0,\sigma_n^2)$ is the additive noise component at the UE. 
Here, we define $\boldsymbol{\zeta} \in \mathbb{C}^{M \times 1}$ as $\boldsymbol{\zeta} = \boldsymbol{\Omega} \boldsymbol{g}$, representing the effective passive beamforming at the transmit terminal of BD-RIS.

\begin{figure}
    \centering
    \includegraphics[width=\columnwidth]{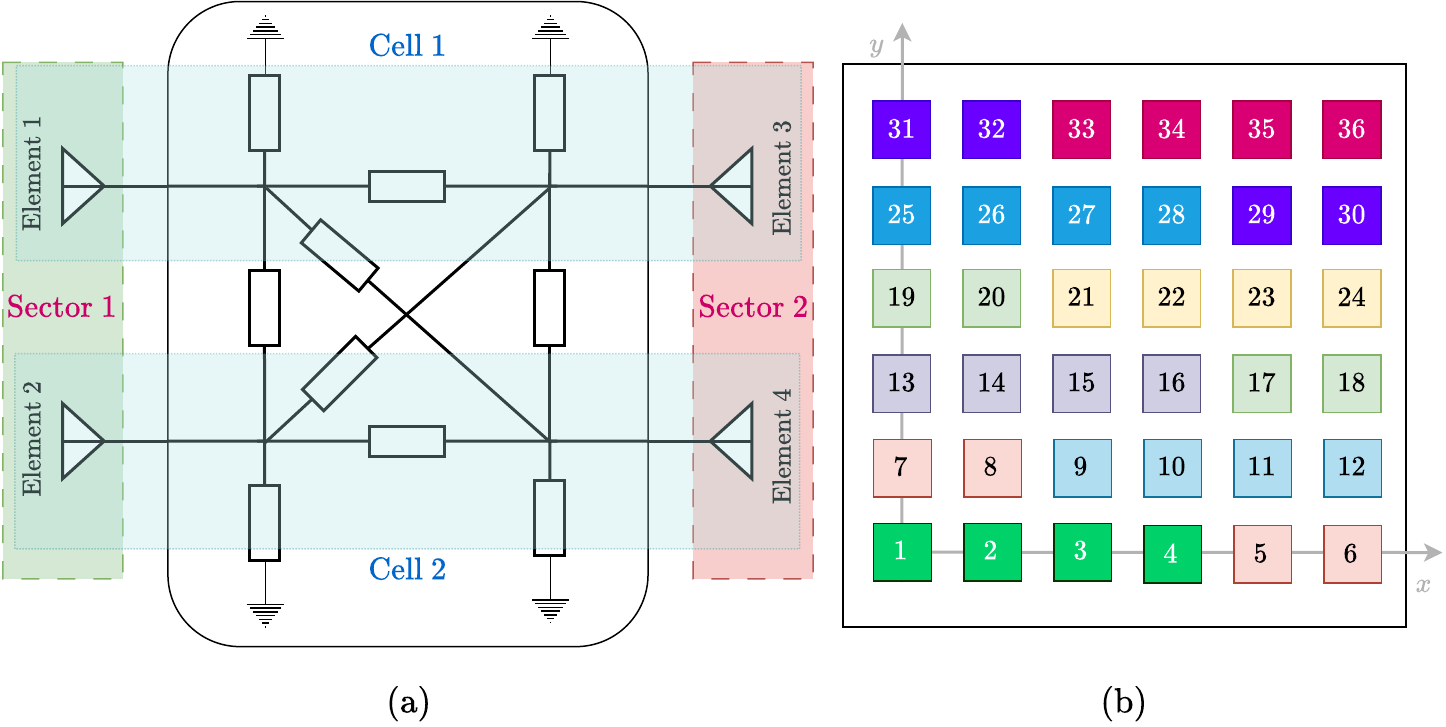}
    \caption{(a) A simplified example of the circuit structure for a fully-connected BD-RIS with two cells, illustrating inter-element connections for amplitude and phase manipulation. (b) A visualization of the group-connected structure with $M = 36$, $G = 9$, and $\bar{M} = 4$, showing the division of cells into groups for balancing performance and circuit complexity.}
    \label{fig:BDRIS_Structure_and_Grouping}
\end{figure}

A maximum-likelihood (ML) detector is employed in the UE to extract the transmitted symbols from the received signal $y$. The detection process is expressed as:
\begin{equation}
    \hat{s} = \arg \min_s | y - \sqrt{P} h_{\textrm{eff}} s |^2,
\end{equation}
where $\hat{s}$ denotes the estimated $\mathcal{M}$-ary symbol and $h_{\textrm{eff}} = \boldsymbol{h}^{\mathsf{T}} \boldsymbol{\Omega} \boldsymbol{g}$ represents the effective SISO channel gain.
It is worth mentioning that channel state information (CSI) can be acquired using existing channel estimation methods developed for SIM and holographic MIMO systems \cite{10445164, 10515204, 10535263, 9716880, 9693928}. \textcolor{black}{Moreover, recent works \cite{10446555, peng2022two} propose lightweight and learning-aided strategies for semi-passive RIS systems, which can be adapted to support efficient CSI acquisition in dynamic environments.}

By following the straightforward steps outlined in \cite{raeisi2023plug,10176315}, the unconditional pairwise error probability (UPEP) can be calculated as follows:
\begin{equation}
    \mathbb{P}(s^{\star} \rightarrow \hat{s}) = \mathbb{E}_{\mathcal{C}} \left [ Q (\frac{|\sqrt{P} \boldsymbol{h}^{\mathsf{T}} \boldsymbol{\Omega} \boldsymbol{g} (s^{\star} - \hat{s})|}{\sqrt{2} \sigma_n}) \right ],
\end{equation}
where $s^{\star}$ denotes the correct symbol considered for transmission, and $\mathcal{C} = \{ \alpha, \varphi_0, \vartheta_0, \beta_{c,\ell}, \varphi_{c,\ell}, \vartheta_{c,\ell} \}, \forall c, \ell$, is the set of channel parameters. Ultimately, by utilizing union-bound approach, we can set a theoretical upper bound on ABER as follows:
\begin{equation}\label{eq:Theoretical Bound}
    \textrm{ABER} \leq \frac{1}{\eta \mathcal{M}} \sum_{s^{\star}} \sum_{\hat{s}} d_H(s^{\star}, \hat{s}) \mathbb{P}(s^{\star} \rightarrow \hat{s}),
\end{equation}
where $d_H(s^{\star}, \hat{s})$ represents the Hamming distance between the binary representations of $s^{\star}$ and $\hat{s}$.

\section{BD-RIS Structure and Configuration}\label{sec:BDRIS_Circuit_Configuration}

In this section, we explore the circuit structure and configuration strategy of the BD-RIS, concentrating on its design to enable efficient passive beamforming.
A simplified example of the circuit structure for the fully-connected BD-RIS is presented in \cite{9913356} and illustrated in Fig. \ref{fig:BDRIS_Structure_and_Grouping}(a). As detailed in Section \ref{sec:system_model}, the BD-RIS elements are divided into two sectors: sector 1, facing the active antenna, captures the incident EM wave, while sector 2, facing outward, transmits the directed signal toward the UE. Each pair of back-to-back elements from the two sectors forms a single cell, as depicted in Fig. \ref{fig:BDRIS_Structure_and_Grouping}(a).

\subsection{Group-Connected Architecture and Circuit Complexity}\label{sec:GroupConnected}

To enable manipulation of both the amplitude and phase of the impinging EM wave, inter-element connections are crucial \cite{li2023reconfigurable, 9913356}. Extending the two-cell structure in Fig. \ref{fig:BDRIS_Structure_and_Grouping}(a) to an $M$-cell fully-connected BD-RIS significantly increases circuit complexity. To balance performance and circuit complexity, a $G$-group-connected structure can be adopted by dividing the $M$ cells into $G$ groups, with each group containing $\bar{M} = M / G$ cells. In this configuration, full connectivity is implemented within each group, while no connections exist between different groups \cite{9913356}. For simplicity, we assume that all groups consist of an equal number of cells $\bar{M}$.\footnote{It is worth noting that varying grouping strategies with dynamic group sizes can lead to different trade-offs between performance and circuit complexity \cite{10453384, 10159457}. However, investigating these strategies and their impact on passive beamforming performance lies beyond the scope of this paper and is deferred to future work.}
The set of cell indices for the $q$-th group is defined as $\mathcal{G}_q = \{(q-1)\bar{M} + 1, \dots, q\bar{M}\}$, where each group contains $\bar{M}$ consecutive cells. This grouping approach is referred to as the linear permutation throughout this paper. To aid visualization, Fig. \ref{fig:BDRIS_Structure_and_Grouping}(b) illustrates a BD-RIS configuration with $M = 36$, $G = 9$, and $\bar{M} = 4$.
Two notable special cases arise from this framework:
\begin{itemize}
    \item When $G = M$, the structure corresponds to a traditional transmissive RIS with a diagonal configuration scattering matrix, allowing only phase manipulation of the impinging signal, limiting its passive beamforming capabilities.
    \item Conversely, when $G = 1$, the configuration transitions to a fully-connected BD-RIS, offering the highest degree of freedom for designing the configuration scattering matrix and enabling superior passive beamforming performance.
\end{itemize}

{\color{black}
    \begin{figure}
        \centering
        \includegraphics[width = \columnwidth]{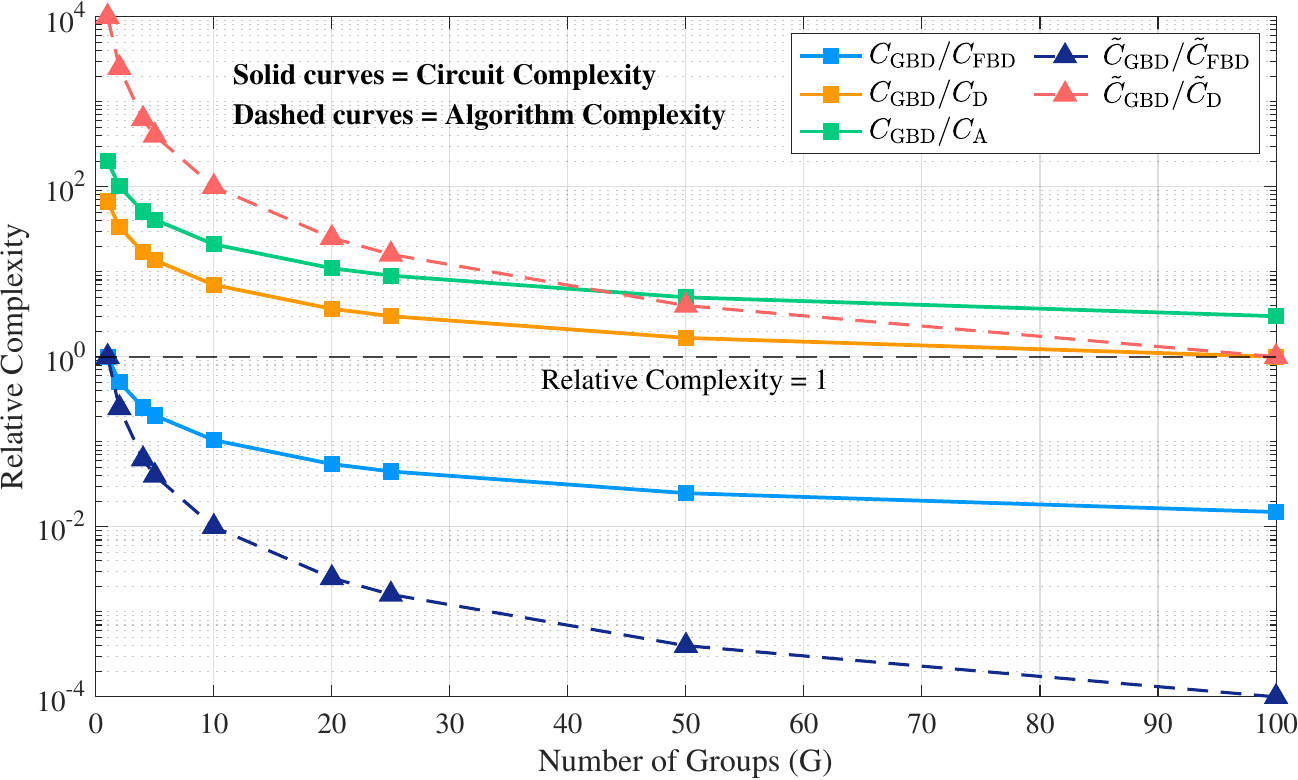}
        \caption{\textcolor{black}{Relative circuit (solid curves) and algorithm (dashed curves) complexity of} the group-connected BD-RIS with respect to fully-connected BD-RIS, D-RIS, and active antenna arrays as a function of the number of groups $G$, with total RIS size fixed at $M = 100$.}
        \label{fig:CircuitComplexity}
    \end{figure}

    \begin{table}[!b] \color{black}
        \centering
        \caption{\textcolor{black}{Circuit Complexity for Different RIS Architectures}}
        \begin{tabular}{l l }
            \toprule
            \textbf{Architecture} & \textbf{Circuit Topology Complexity} \\
            \toprule
            Active Antenna Array & $C_{\text{A}} = M$ \\
            D-RIS & $C_{\text{D}} = 3M$ \\
            BD-RIS (Fully-Connected) & $C_{\text{FBD}} = (2M + 1)M$ \\
            BD-RIS (Group-Connected) & $C_{\text{GBD}} = \left(2M/G + 1\right)M$ \\
            \bottomrule
        \end{tabular}
        \label{Tab:Circuit Complexity}
    \end{table}   

    Table \ref{Tab:Circuit Complexity} summarizes the circuit topology complexity associated with various antenna array architectures, including active antenna arrays equipped with a single phase shifter per antenna, traditional D-RIS structures, fully-connected BD-RIS configurations, and group-connected BD-RIS designs, as discussed in \cite{10530995, 9913356}. To offer a clearer understanding of how circuit complexity scales across these architectures, \textcolor{black}{the solid curves in} Fig. \ref{fig:CircuitComplexity} illustrate the relative complexity of the group-connected BD-RIS as a function of the number of groups $G$, with the RIS size fixed at $M = 100$. As shown in Fig. \ref{fig:CircuitComplexity}, the relative complexity with respect to the fully-connected BD-RIS, denoted $C_{\mathrm{GBD}}/C_{\mathrm{FBD}}$, decreases notably as $G$ increases, dropping to approximately $0.1$ for $G = 10$ and below $0.05$ for $G = 20$. Even at $G = 2$, the group-connected architecture achieves a significant complexity reduction of about $50\%$ compared to its fully-connected counterpart.

    Moreover, the curves for $C_{\mathrm{GBD}}/C_{\mathrm{D}}$ and $C_{\mathrm{GBD}}/C_{\mathrm{A}}$ remain above $1$ across all $G$, indicating that while group-connected BD-RIS significantly reduces internal circuit complexity, it still entails higher circuit complexity compared to D-RIS and active analog arrays.\footnote{\color{black} It is worth noting that the proposed grouping strategy also contributes to reducing control complexity. Specifically, by decreasing the number of configured elements, the grouping approach lowers the computational and signaling burden required for BD-RIS configuration during operation.}
}

\subsection{\color{black}BD-RIS Configuration and Precoding Assumptions}\label{Sec: BD-RIS Config}

The configuration scattering matrix of the $q$-th group, $\boldsymbol{\Omega}_q \in \mathbb{C}^{\bar{M} \times \bar{M}}$, is a full matrix that satisfies the unitary constraint $\boldsymbol{\Omega}_q^\mathsf{H} \boldsymbol{\Omega}_q = \boldsymbol{I}_{\bar{M}}$,
ensuring that each group operates without power loss.{\color{black} \footnote{\color{black} In line with general BD-RIS modeling frameworks \cite{10530995, 9913356}, we do not impose the symmetry condition, as it is not required for transmissive BD-RIS configurations. Nevertheless, incorporating symmetry may still offer practical advantages in implementation, such as reducing the number of control lines by enabling joint control of symmetric impedance pairs.}}
The BS maintains direct access to the BD-RIS, enabling the transmission of control signals to an embedded controller for configuring the BD-RIS. To obtain BD-RIS scattering matrix, we employ Takagi’s decomposition, as outlined in \cite{10187688}. Hence, utilizing the BS-RIS channel $\boldsymbol{g}$, the effective beamforming vector $\boldsymbol{b} \in \mathbb{C}^{M\times1}$ corresponding to the RIS-UE channel, and employing singular value decomposition (SVD) technique, the BD-RIS configuration for the general case of $G$ groups is detailed in Algorithm \ref{Alg:BDRIS_Configuration}.

\begin{algorithm}
\caption{BD-RIS Configuration for Linear Permutation.}
\label{Alg:BDRIS_Configuration}
\begin{algorithmic}[1] 
    \State Determine the beamforming vector $\boldsymbol{b}$ from the RIS-UE channel, utilizing one of the possible approaches, such as those outlined in this Section.
    \For{each group of cells $q \in \{ 1,\dots,G \}$}
        \State Constitute $\boldsymbol{a}_q = \boldsymbol{b}_{[(q-1)\bar{M} + 1: q\bar{M}]}$. 
        
        \State Constitute $\boldsymbol{g}_q = \boldsymbol{g}_{[(q-1)\bar{M} + 1: q\bar{M}]}$.

        \State Calculate $\boldsymbol{u}_q = \boldsymbol{g}_q / \lVert \boldsymbol{g}_q \rVert$. 

        \State Compute $\boldsymbol{v}_{q} = \boldsymbol{a}_q / \lVert \boldsymbol{a}_q \rVert$. 

        \State Derive the symmetric matrix $\boldsymbol{A}_q = \boldsymbol{v}_{q} \boldsymbol{u}_q^\mathsf{H} + (\boldsymbol{v}_{q} \boldsymbol{u}_q^\mathsf{H})^\mathsf{T}$. 
        
        \State Decompose $\boldsymbol{A}_q$ using SVD as $\boldsymbol{A}_q = \boldsymbol{U}_q \boldsymbol{\Sigma}_q \boldsymbol{V}_q^\mathsf{H}$. 
        
        \State Calculate $\boldsymbol{\nu}_q = \text{diag}(\boldsymbol{U}_q^\mathsf{H} \boldsymbol{V}_q^*)$.
        
        \State Determine $\boldsymbol{\rho}_q = \angle \boldsymbol{\nu}_q / 2$.

        \State Compute $\boldsymbol{\mathcal{Q}}_q = \boldsymbol{U}_q \, \text{diag}(e^{(\jmath \boldsymbol{\rho}_q)})$. 
        
        \State Takagi's decomposition: $\boldsymbol{A}_q = \boldsymbol{\mathcal{Q}}_q \boldsymbol{\Sigma}_q \boldsymbol{\mathcal{Q}}_q^\mathsf{T}$.
        
        \State Scattering matrix for the $q$-th group: $\boldsymbol{\Omega}_q = \boldsymbol{\mathcal{Q}}_q\boldsymbol{\mathcal{Q}}_q^\mathsf{T}$.
    \EndFor
    \State Assemble the overall BD-RIS scattering matrix as $\boldsymbol{\Omega} = \blkdiag(\boldsymbol{\Omega}_1,\dots,\boldsymbol{\Omega}_G)$.
\end{algorithmic}
\end{algorithm}

\textcolor{black}{The computational complexity of Algorithm \ref{Alg:BDRIS_Configuration} can be analyzed per group of size $\bar{M}$. Vector normalization, phase extraction, and diagonal operations scale linearly or quadratically with $\bar{M}$, while cubic terms arise from the SVD, Takagi’s decomposition, and matrix multiplications. Since the procedure is repeated across all $G$ groups, the overall complexity is on the order of $\mathcal{O}(G\bar{M}^3) = \mathcal{O}(M^3 / G^2)$. This indicates that Algorithm \ref{Alg:BDRIS_Configuration} scales cubically with the total number of RIS elements $M$, but decreases quadratically with the number of groups $G$. 
The dashed curves in Fig.~\ref{fig:CircuitComplexity} illustrate the relative algorithm complexity of the group-connected BD-RIS with respect to the fully-connected BD-RIS, i.e., $\tilde{C}_{\mathrm{GBD}}/\tilde{C}_{\mathrm{FBD}}$, and with respect to the diagonal RIS, i.e., $\tilde{C}_{\mathrm{GBD}}/\tilde{C}_{\mathrm{D}}$, for a fixed number of RIS elements $M = 100$. Compared to circuit complexity, it is evident that the group-connecting strategy achieves a more substantial reduction in algorithm complexity.}

A critical factor influencing the BD-RIS configuration and overall system performance is the stochastic nature of the RIS-UE channel. To ensure effective performance, the BS must access the efficient beamforming vector $\boldsymbol{b} \in \mathbb{C}^{M\times1}$ which aligns with the RIS-UE channel during each time block. Obviously, the availability of CSI (full or partial) is crucial for deriving such an effective beamforming vector. In this paper, we evaluate the performance of the proposed system under three different cases based on CSI availability:
\begin{enumerate}
    \item \textit{Case 1 (Full CSI is available at the BS)}. The BS, equipped with a large antenna array, can estimate the RIS-UE channel by processing pilot signals transmitted by the UE during the uplink phase \cite{10445164, 9716880, 9693928}. The optimal \textit{unconstrained} beamforming vector for the channel $\boldsymbol{h}$ is determined as $\boldsymbol{b} = \boldsymbol{v}_1$, where $\boldsymbol{v}_1 \in \mathbb{C}^{M\times1}$ represents the first column of the unitary matrix $\boldsymbol{V}~\in~\mathbb{C}^{M\times M}$ \cite{el2014spatially}. This matrix is obtained via the ordered SVD of the channel, expressed as $\boldsymbol{h}^{\mathsf{T}} =  \boldsymbol{\Sigma} \boldsymbol{V}^{\mathsf{H}}$, where $\boldsymbol{\Sigma}$ is a $1 \times M$ vector of singular values with non-negative elements. The beamforming vector $\boldsymbol{v}_1$ encapsulates the optimal complex weights—comprising both phase and magnitude—for the BD-RIS elements to maximize beamforming gain. However, in systems with constant modulus constraints, such as BS-integrated D-RIS or active analog beamforming antenna arrays utilizing analog phase shifter networks, only the phase information from $\boldsymbol{v}_1$ can be exploited.
    \item \textit{Case 2 (Full CSI is available at the UE)}. In general, the UE may handle channel estimation;\footnote{Although this work focuses on an emulated MISO system, extending it to an emulated MIMO system with an additional BD-RIS integrated at the UE is straightforward. In such scenario, the UE would also be capable of performing channel estimation.} however, feeding back the full CSI to the BS, particularly in systems with massive MIMO, is often impractical and can result in significant overhead \cite{el2014spatially}. To address this, the UE can utilize the spatially sparse precoding (SSP) scheme proposed in \cite{el2014spatially}. With this approach, the UE identifies the best match for the channel’s dominant eigenmode $\boldsymbol{v}_1$ from a predefined codebook. As the codebook is shared between the BS and UE, the UE only needs to transmit the index of the selected codeword, significantly reducing feedback overhead. In this paper, we utilize a codebook consisting of beams corresponding to equally spaced angles \cite{el2014spatially}, defined as follows:
    \begin{equation}
        \mathcal{B} = \{ \boldsymbol{a}(\phi,\theta) | \phi \in [-\pi:\frac{\pi}{36}:\pi], \theta \in [0:\frac{\pi}{36}:\frac{\pi}{2}] \}.
    \end{equation}
    The optimal beam alignment with $\boldsymbol{v}_1$ is determined by identifying the beam index that maximizes the correlation, computed as:
    \begin{equation}
        \{ i^{\star},j^{\star} \} = \arg \max_{i,j} \ | \boldsymbol{a}^{\mathsf{H}}(\phi_i,\theta_j) \boldsymbol{v}_1 |,
    \end{equation}
    where $\phi_i$ and $\theta_j$ represent the $i$-th and $j$-th quantized azimuth and elevation angles within the ranges $\phi = [-\pi:\frac{\pi}{36}:\pi]$ and $\theta = [0:\frac{\pi}{36}:\frac{\pi}{2}]$, respectively. Subsequently, the UE transmits the indices $\{ i^{\star}, j^{\star} \}$ to the BS via a limited feedback link, allowing the BS to compute the beamforming vector as $\boldsymbol{b} = \boldsymbol{a}(\phi_{i^{\star}},\theta_{j^{\star}})$.
    \item \textit{Case 3 (Partial CSI is available)}. In certain scenarios, full CSI may be unavailable, or obtaining partial CSI is preferred to minimize channel estimation overhead \cite{8891298, koc2020hybrid}. Angular channel information can be extracted using AOD estimation algorithms \cite{9086460}, allowing the identification of an effective transmission direction.\footnote{In order to emulate this process, we search across the available LOS and NLOS directions to maximize the transmitted signal in the target sub-channel while minimizing the leakage power into unintended sub-channels. Mathematically, the optimal direction $c^{\star}$ is determined by solving $c^{\star} = \arg \max_{c = 0, \dots, C } | \boldsymbol{h}_c^{\mathsf{T}} \boldsymbol{a}(\varphi_c, \vartheta_c) |/|(\sum_{i \neq c}  \boldsymbol{h}_i^{\mathsf{T}})\boldsymbol{a}(\varphi_c, \vartheta_c)|$; hence, the effective beamforming vector can be obtained as $\boldsymbol{b} = \boldsymbol{a}(\varphi_{c^{\star}},\vartheta_{c^{\star}})$. It is worth noting that the direction of the $c$-th NLOS cluster is assumed to align with its best effective path, as defined in \cite[equation (13)]{10176315}.} Subsequently, the effective channel $h_{\textrm{eff}}$, required for detection, can be readily estimated.
\end{enumerate}

\section{Channel Amplitude Variation: A Key Design Metric for BS-Integrated RIS Performance}\label{sec:CAV}
In this section, we take a closer look at the design parameters of the BS-integrated BD-RIS and analyze how these parameters impact system performance compared to traditional structures. Specifically, we focus on the BS-RIS channel, a deterministic component of the BS’s internal architecture, shaped by the BS’s design parameters.

According to the Rayleigh-Sommerfeld diffraction theory for near-field propagation, as presented in (\ref{eq:Rayleigh-Sommerfeld}), the amplitude of the channel coefficient corresponding to each element is inversely proportional to $d_m^3$, emphasizing a rapid decay in channel gain amplitude with increasing distance. 
Unlike traditional D-RIS, which is limited to phase compensation, BD-RIS can adjust both phase and amplitude variations. 
Consequently, when the signal emitted by the active antenna reaches the BD-RIS, it achieves alignment in both phase and amplitude across RIS elements. The primary performance degradation in BS-integrated BD-RIS compared to active antenna arrays arises from the multiplicative path loss due to the BS-RIS distance.
In contrast, D-RIS, with its phase-only adjustment capability, cannot align the amplitude variations across RIS elements. This limitation introduces additional performance degradation beyond the multiplicative path loss, as amplitude misalignment further reduces system efficiency.
This enhanced capability of BD-RIS provides greater flexibility in beamforming design, whereas the performance of D-RIS remains fundamentally constrained by its limited adjustment capabilities.

Since amplitude variations across RIS elements play a crucial role, we focus on this aspect in the remainder of this section. As previously mentioned, the amplitude of the channel gain corresponding to the $m$-th element is inversely proportional to $d_m^3$. The variations in $d_m$ are significantly affected by two key factors: the array aperture and the BS-RIS separation ($d_c$). The array aperture is directly proportional to the number of RIS elements; larger RIS arrays incorporate more elements, resulting in greater variations in $d_m$ and, consequently, more pronounced discrepancies in channel amplitudes across the elements. In contrast, as the BS-RIS separation ($d_c$) increases, the relative variations in $d_m$ decrease, leading to reduced differences in channel amplitudes across RIS elements.

\begin{figure}
    \centering
    \includegraphics[scale = 0.35]{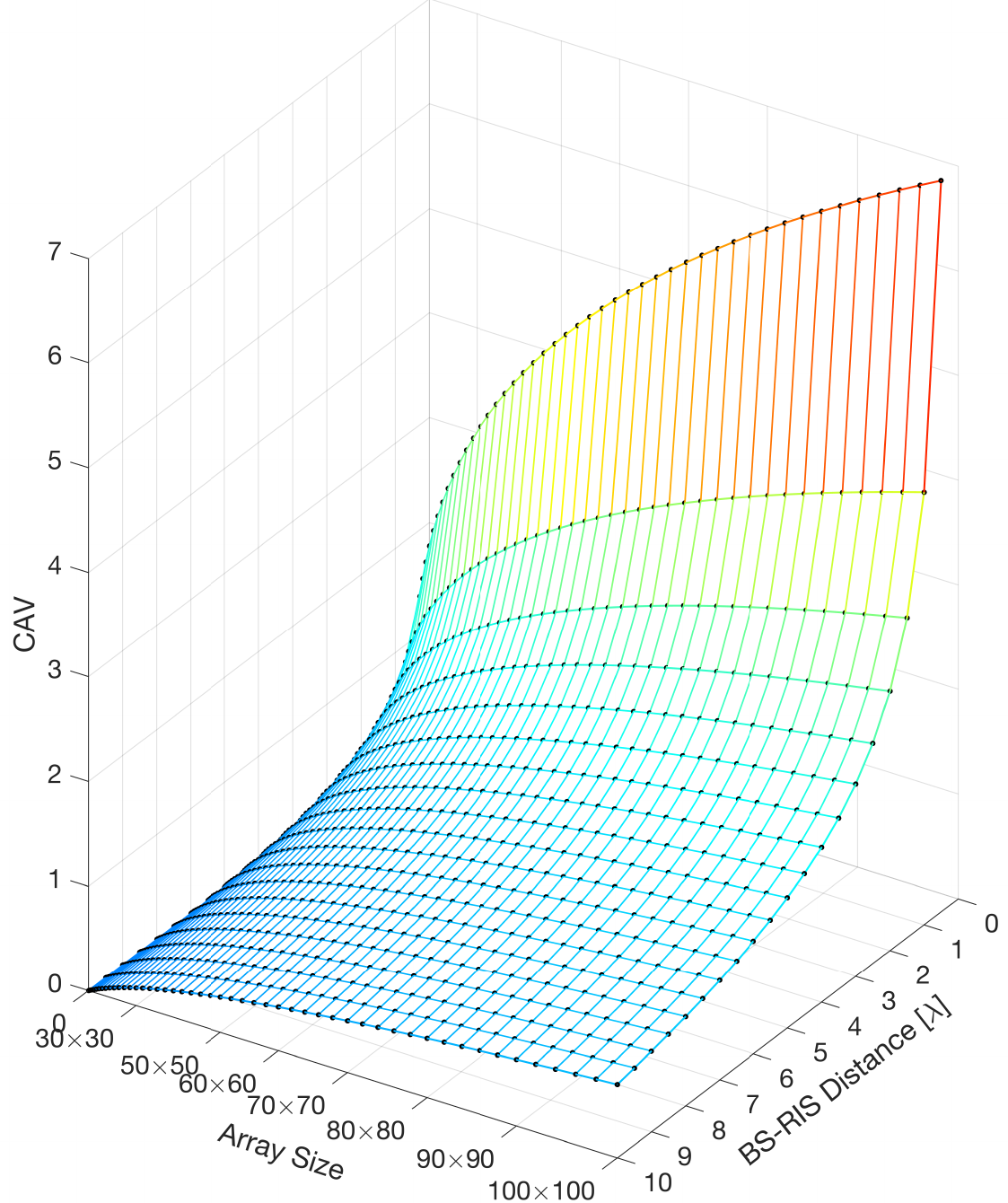}
    \caption{CAV for different array sizes and BS-RIS separations.}
    \label{fig:CAV}
\end{figure}

To quantify and visualize the relative spread of channel amplitude variations between the active antenna and the RIS elements, we introduce a novel metric called \textit{channel amplitude variation} (CAV). CAV quantifies the variability in channel amplitudes relative to their mean and is specifically defined for the BS-RIS channel vector $\boldsymbol{g}$. Mathematically, CAV is computed as the standard deviation of the channel amplitude components normalized by their mean, expressed as
\begin{equation}\label{eq:CAV}
    \mathrm{CAV} = \frac{\color{black} \sigma_{|\boldsymbol{g}|}}{\mu_{|\boldsymbol{g}|}},
\end{equation}
\textcolor{black}{where $\sigma_{|\boldsymbol{g}|}$ denotes the standard deviation of the amplitudes of the channel gains, formally expressed as
\begin{equation}\label{eq:std (CAV)}
    \sigma_{|\boldsymbol{g}|} = \sqrt{\frac{1}{M}\sum_{m = 1}^{M} (|[\boldsymbol{g}]_m| - \mu_{|\boldsymbol{g}|})^2},
\end{equation}
and $\mu_{|\boldsymbol{g}|}$ is the mean amplitude of the channel gains, given by}
\begin{equation}\label{eq:mean (CAV)}
    \mu_{|\boldsymbol{g}|} = \frac{1}{M} \sum_{m = 1}^M |[\boldsymbol{g}]_m|.
\end{equation}
The normalization by the mean ensures that the CAV metric provides a dimensionless and scale-invariant representation of amplitude variability, making it suitable for comparisons across different system configurations and BS-RIS separations.

Fig. \ref{fig:CAV} depicts the CAV as a function of both the array size and the BS-RIS separation ($d_c$). As discussed, the CAV increases with larger array sizes and smaller $d_c$. Note that the primary motivation for integrating RIS within the BS is to enable the implementation of larger arrays to achieve narrower beams, thereby reducing cost and power consumption. Additionally, minimizing the BS-RIS separation ($d_c$) is crucial to mitigate the impact of multiplicative path loss. Consequently, designing a system that aligns with these two characteristics is critical, underscoring the necessity of adopting BD-RIS to effectively manage amplitude variations.
Notably, as $d_c$ decreases, placing the RIS closer to the active antenna, the CAV exhibits greater sensitivity, further underlining the importance of advanced designs like BD-RIS to mitigate these effects.

{\color{black}
In the following, we analyze the received SNR of both BD-RIS and D-RIS integrated systems to theoretically characterize the performance gain of BD-RIS over D-RIS. Additional discussion can be found in \cite{9514409}. From the signal model in (\ref{eq:signal model}), the SNR is proportional to $|\boldsymbol{h}^\mathsf{T} \boldsymbol{\Omega} \boldsymbol{g}|^2$. 
This expression can be specialized for different RIS architectures as follows:

\begin{enumerate}[label=\textit{\roman*})]
    \item \textit{D-RIS, phase-only adjustment.} In this case, the optimal scattering matrix is given by
    \begin{equation}
        \boldsymbol{\Omega}^* = \diag (e^{\jmath \omega_1^*}, e^{\jmath \omega_2^*}, \dots, e^{\jmath \omega_M^*})
    \end{equation}
    where $\omega^*_m = - \arg ([\boldsymbol{h}]_m[\boldsymbol{g}]_m)$. After aligning the phases through the RIS, the resulting SNR is proportional to
    \begin{equation}\label{eq:SNR D-RIS}
        \mathrm{SNR}_{\mathrm{D-RIS}} \propto \left( \sum_{m=1}^M |[\boldsymbol{h}]_m | \,|[\boldsymbol{g}]_m| \right)^2.
    \end{equation}

    \item \textit{BD-RIS, joint phase-amplitude adjustment.} In this case, applying the Cauchy–Schwarz inequality yields
    \begin{equation}\label{eq:Cauchy-Schwarz}
        |\boldsymbol{h}^{\mathsf{T}} \boldsymbol{\Omega} \boldsymbol{g}|^2 \leq \| \boldsymbol{h} \|^2 \| \boldsymbol{\Omega} \boldsymbol{g} \|^2 \overset{a}{=} \| \boldsymbol{h} \|^2 \| \boldsymbol{g} \|^2,
    \end{equation}
    where $a$ follows from the unitary property of the scattering matrix, i.e., $\boldsymbol{\Omega}^\mathsf{H} \boldsymbol{\Omega} = \boldsymbol{I}$. From (\ref{eq:Cauchy-Schwarz}), the maximum SNR is achieved with appropriate selection of $\boldsymbol{\Omega}$. Therefore, the maximum SNR in this case is proportional to
    \begin{equation}\label{eq: SNR BD-RIS}
        \mathrm{SNR}_{\mathrm{BD-RIS}} \propto \| \boldsymbol{h} \|^2 \| \boldsymbol{g} \|^2.
    \end{equation}
\end{enumerate}

\begin{proposition}\label{prop:gainfloor}
For the proposed BS-integrated BD-RIS, the minimum guaranteed SNR gain of BD-RIS over D-RIS is
\begin{equation}
G_{\min} = 10\log_{10}\!\left(1 + \mathrm{CAV}^2\right).
\label{eq:Gmin}
\end{equation}
\end{proposition}
\begin{proof}
See Appendix~\ref{app:gainfloor}.
\end{proof}

\begin{corollary}
The SNR gain floor in (\ref{eq:Gmin}) arises when the RIS-UE channel consists of a single path, where the RIS-UE channel coefficients associated with all RIS elements have identical amplitudes. In this case, phase alignment alone is sufficient to compensate for the phase variations of the RIS-UE channel. Consequently, the BD-RIS gain over D-RIS is confined to compensating for CAV.
\end{corollary}

\begin{proposition}\label{prop:asymptotic}
For the proposed BS-integrated BD-RIS, when the RIS-UE channel is subject to rich scattering (i.e., Rayleigh fading), the SNR gain of BD-RIS over D-RIS asymptotically converges to
\begin{equation}\label{eq:Gmax}
G_{\max} \;=\; G_{min}+\; 10\log_{10}\!\Big(\frac{4}{\pi}\Big).
\end{equation}
\end{proposition}

\begin{proof}
See Appendix~\ref{app:asymptotic}.
\end{proof}

\begin{corollary}\label{cor:asymptotic}
The asymptotic gain in (\ref{eq:Gmax}) arises when the RIS-UE channel experiences rich scattering (i.e., Rayleigh fading). In this case, BD-RIS provides an additional fixed gain of $10\log_{10}(4/\pi)\!\approx\!1.05$ dB beyond the floor in Proposition~\ref{prop:gainfloor}, stemming from its ability to compensate for amplitude diversity in the RIS-UE channel.
\end{corollary}

}

{
\color{black}
Ultimately, as discussed in Section \ref{sec:GroupConnected}, adopting a group-connected architecture offers an effective means to reduce the BD-RIS circuit complexity. Propositions \ref{prop:gainfloor} and \ref{prop:asymptotic} have already established that both the floor and ceiling performances are governed by the $\mathrm{CAV}$. Consequently, as long as a group-connected strategy preserves the $\mathrm{CAV}$, these bounds remain unaffected despite modifications to the circuit structure. The next proposition formalizes this by showing that a symmetric grouping strategy ensures the $\mathrm{CAV}$ remains unchanged.

\begin{proposition}  \label{prop:SymmetricGeometryCAV}
    Let  $A = \{ |[\boldsymbol{g}]_m| : m=1,\dots,M \}$ be the set of channel amplitudes between the active antenna and the RIS elements. Suppose that $A$ can be partitioned into $K$ disjoint symmetric subsets $A_1, A_2, \dots, A_K$, where each subset contains entries (corresponding to different RIS element indices) having identical numerical values. Then, the CAV defined in (\ref{eq:CAV}) satisfies  
    \begin{equation}
        \mathrm{CAV}(A) = \mathrm{CAV}(A_1) = \cdots = \mathrm{CAV}(A_K).
    \end{equation}
\end{proposition}

\begin{proof} 
    See Appendix \ref{app:SymmetricGeometryCAV}.
\end{proof}
}

\section{Simulation Results and System Evaluation}\label{Sec:SimulationResutls}
This section presents the computer simulation results for the proposed BS-integrated BD-RIS, offering a comparative evaluation against considered benchmarks. The simulation setup is first outlined for a realistic mmWave street canyon environment.
Next, the beamforming performance is assessed, demonstrating the potential of the BD-RIS to outperform the traditional D-RIS structure and closely rival the performance of active analog beamforming antenna arrays. ABER and achievable rate analyses are conducted under diverse circumstances to comprehensively evaluate the system’s performance from multiple perspectives. Finally, it is shown that dividing the BD-RIS cells into groups according to their geometric symmetry around the active antenna allows for a reduction in circuit complexity while maintaining system performance.

\subsection{Simulation Setup}\label{sec:simulaion_setup}
We consider a P2P scenario in the street canyon environment where the distance between the BS-integrated BD-RIS and UE is considered $d = 20$ m \cite{akdeniz2014millimeter}. We consider a carrier frequency of $f_c = 28$ GHz; hence, we have a mmWave environment with a spatially sparse-scattered channel described in (\ref{eq:ChannelModel}).
Due to the highly dynamic nature of urban environments and the susceptibility of mmWave signals to blockages, the presence of a LOS path cannot be consistently guaranteed. To comprehensively evaluate the system’s performance under varying conditions, we consider two distinct scenarios:
\begin{itemize}
    \item \textit{Scenario 1}. A dominant LOS component is present alongside several NLOS clusters.

    \item \textit{Scenario 2}. The LOS is blocked, and communication relies on NLOS clusters, ensuring connectivity \textcolor{black}{in challenging scenarios, such as complex urban environments.}
\end{itemize}


\begin{table}[t]
    \centering
    \caption{Simulation Parameters.}
    \begin{tabular}{c c c}
        \toprule
         \textbf{Parameter} & \textbf{Description} & \textbf{Value} \\
        \toprule
        $f_c$ & Carrier frequency & $28$ GHz  \\
        $B$ & Communication bandwidth & $100$ MHz \textcolor{black}{\cite{10176315, raeisi2023plug}} \\
        $d$ & Communication range & $20$ m \textcolor{black}{\cite{akdeniz2014millimeter}} \\
        $d_x = d_y$ & RIS inter-element spacing 
        & $\lambda/2$ 
        \\
        $d_c$ & BS-RIS separation & $\lambda/2$ \\
        $A$ & Area of each RIS elements & $(\lambda/2)^2$  \\
        $M$ & BD-RIS size & $10 \times 10$ \\
        $C$ & Number of clusters & $8$ \textcolor{black}{\cite{10176315, el2014spatially}}  \\
        $L$ & Number of paths per cluster & $10$ \textcolor{black}{\cite{10176315, el2014spatially}}\\
        $\varsigma_c$ & Angular spread in $c$-th cluster & $7.5^{\circ}$ \textcolor{black}{\cite{10176315, el2014spatially}}\\
        $P$ & BS transmitted power & $20$ dBm \\
        $\mathcal{M}$ & $\mathcal{M}$-ary constellation order & $2$ \\
        $a_{\textrm{LOS}}$ & Reference LOS path loss & $61.4$ dB \textcolor{black}{\cite{akdeniz2014millimeter}} \\
        $a_{\textrm{NLOS}}$ & Reference NLOS path loss & $72$ dB \textcolor{black}{\cite{akdeniz2014millimeter}} \\
        $b_{\textrm{LOS}}$ & LOS path loss exponent & $2$ \textcolor{black}{\cite{akdeniz2014millimeter}} \\
        $b_{\textrm{NLOS}}$ & NLOS path loss exponent & $2.92$ \textcolor{black}{\cite{akdeniz2014millimeter}} \\
        $\sigma_{\xi,\textrm{LOS}}$ & LOS shadow fading severity & $5.8$ dB \textcolor{black}{\cite{akdeniz2014millimeter}} \\
        $\sigma_{\xi,\textrm{NLOS}}$ & NLOS shadow fading severity & $8.7$ dB \textcolor{black}{\cite{akdeniz2014millimeter}} \\
        Noise PSD & Noise power spectral density & $-174$ dBm/Hz \textcolor{black}{\cite{10176315, raeisi2023plug}} \\
        \bottomrule
    \end{tabular}
    \label{tab:SimulationParameters}
\end{table}

The path loss model for $i \in \{ \textrm{LOS}, \textrm{NLOS} \}$ path is defined as follows \cite{akdeniz2014millimeter}:
\begin{equation}\label{eq:PathLoss}
\textrm{PL}_i(d) \, [\mathrm{dB}] = a_i + 10 \, b_i \log_{10}(d) + \xi_i,
\end{equation}
where $a_i$ represents the reference path loss at the reference distance, $b_i$ is the path loss exponent that determines the rate of signal attenuation with distance, and $\xi_i \sim \mathcal{N}(0, \sigma_{\xi,i}^2)$ models the effect of shadowing, with $\sigma_{\xi,i}$ indicating the shadow fading severity.
\textcolor{black}{The parameters $a_i$, $b_i$, and $\sigma_{\xi,i}$, summarized in Table \ref{tab:SimulationParameters}, are selected based on the experimental campaign conducted in a dense urban environment in New York City \textcolor{black}{\cite{akdeniz2014millimeter}}.}
The noise power spectral density (PSD) is assumed to be $-174 \, \mathrm{dBm/Hz}$ \textcolor{black}{\cite{10176315, raeisi2023plug}}, and the system bandwidth (BW) is set to $B = 100 \, \mathrm{MHz}$ \textcolor{black}{\cite{10176315, raeisi2023plug}}. Consequently, the noise power is calculated as $\sigma_n^2 = \mathrm{PSD} + 10 \log_{10}(B) = -94 \, \mathrm{dBm}$. 
The number of NLOS clusters and paths is set to $C = 8$ and $L = 10$, respectively, following the parameters in \textcolor{black}{\cite{10176315, el2014spatially}}. The angular spread of the NLOS clusters is assumed to be $\varsigma_c = 7.5^\circ$ as specified in \textcolor{black}{\cite{10176315}}. Additionally, a binary phase shift keying (BPSK) signaling scheme is employed for ABER analysis. Unless specified otherwise, the default BD-RIS configuration in all simulations is fully-connected, showcasing the maximum potential of BD-RIS.
A summary of the default system parameter values\textcolor{black}{, carefully selected based on established works in the literature,} is presented in Table \ref{tab:SimulationParameters}.

To assess the performance of the proposed BS-integrated BD-RIS, we consider two benchmarks as outlined below:
\begin{itemize}
    \item \textit{Benchmark 1 (D-RIS)}. A traditional transmissive RIS integrated at the BS \cite{9598898, 9685418, 9991837, 10694297, 10158690, 10515204, 10679332, 10643881, 10622963, 10557708, 9913356} features a diagonal phase shift matrix as its defining characteristic. \textcolor{black}{This architecture is functionally equivalent to a transmission-only STAR-RIS configuration \cite{9570143} and also corresponds to a single-layer SIM \cite{10515204}.} In this benchmark, the phase shift matrix is given by $\boldsymbol{\Omega} = \diag \left( e^{-\jmath \arg (\boldsymbol{g} \odot \boldsymbol{b}^{*})} \right) $.

    \item \textit{Benchmark 2 (Active Analog Beamforming Antenna Array)}.  The BS is configured with a planar active analog beamforming antenna array comprising $M_x \times M_y$ elements connected to a network of analog single phase shifters \cite{yu2018hardware}. In this configuration, the analog beamforming vector is directly applied to the phase shift network. It is important to note that the analog phase shift network imposes a constant modulus constraint on the beamforming weights. Consequently, when employing SVD, a phase-only beamforming vector is used, calculated as $\boldsymbol{b} = e^{\jmath \arg(\boldsymbol{v}_1)}$. The signal model corresponding to this benchmark can be expressed as:
    \begin{equation}
        y = \sqrt{P} \boldsymbol{h}^{\mathsf{T}} \boldsymbol{b} s + n.
    \end{equation}
\end{itemize}

\begin{figure}
    \centering
    \includegraphics[width=\columnwidth]{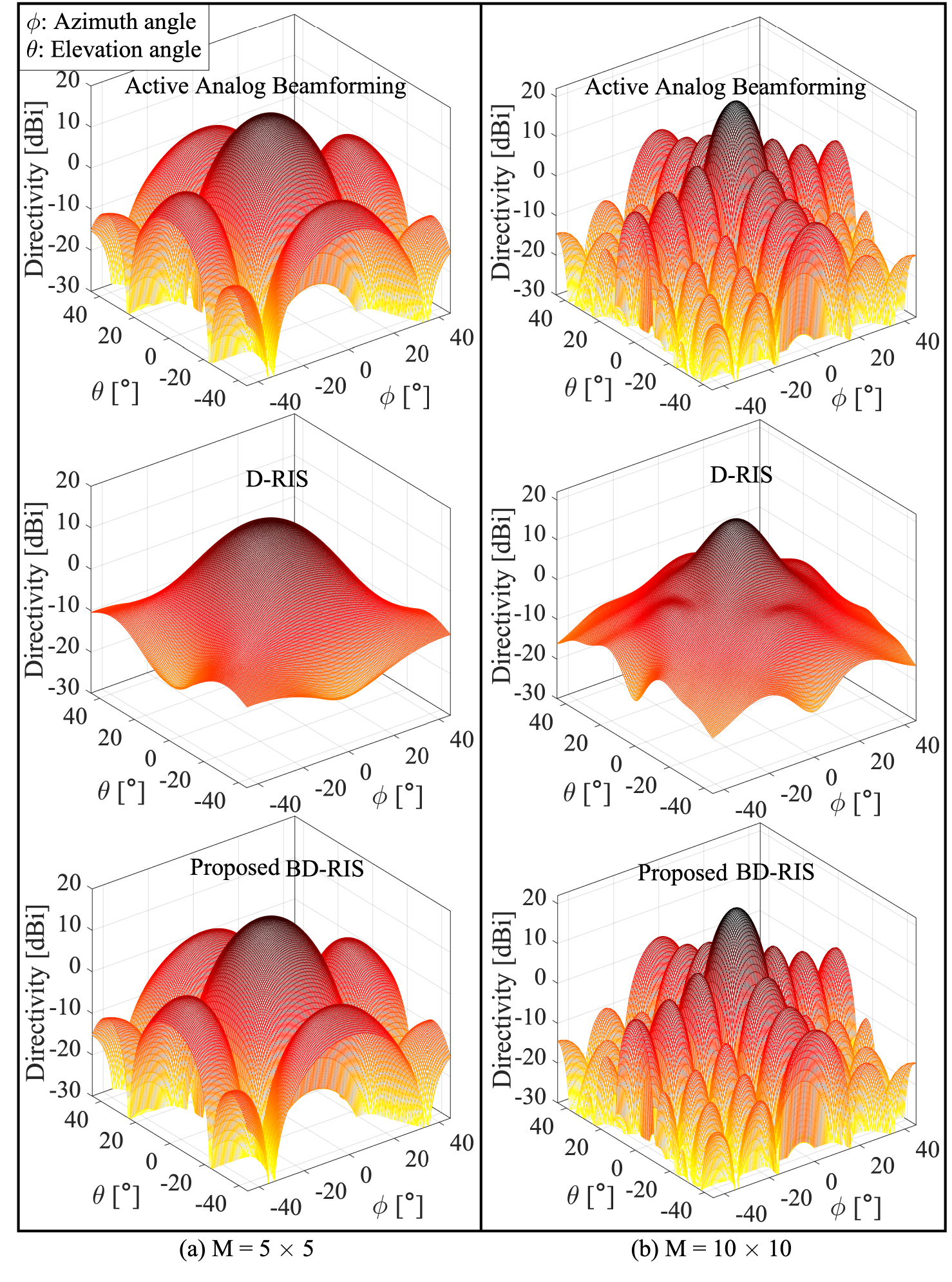}
    \caption{Beamforming performance of different array structures, \textcolor{black}{highlighting the identical beam pattern of the BD-RIS and active analog beamforming, in contrast to the} \textcolor{black}{notably degraded beamforming gain of the D-RIS:} (a) beam patterns for $M = 5 \times 5$ array, (b) beam patterns for $M = 10 \times 10$ array.}
    \label{fig:Beamforming_ArraySize}
\end{figure}

\begin{table}[!b]
    \centering
    \caption{\scriptsize Beamforming Performance Metrics for Different BS Structures}
    \resizebox{\columnwidth}{!}{%
        \begin{tabular}{|c|c|c|c|c|c|c|}
            \hline
            \multirow{2}{*}{} & \multicolumn{2}{c|}{\textbf{Active Array}} & \multicolumn{2}{c|}{\textbf{D-RIS}} & \multicolumn{2}{c|}{\textbf{BD-RIS}} \\ \cline{2-7}
                              & $\boldsymbol{5 \times 5}$ & $\boldsymbol{10 \times 10}$ & $\boldsymbol{5 \times 5}$ & $\boldsymbol{10 \times 10}$ & $\boldsymbol{5 \times 5}$ & $\boldsymbol{10 \times 10}$ \\ \hline
            \textbf{PPD [dBi]} & 15.198 & 21.548 & 13.577 & 17.602 & 15.198 & 21.548 \\ \hline
            \textbf{HPPD [dBi]} & 12.119 & 18.669 & 10.422 & 14.645 & 12.119 & 18.669 \\ \hline
            \textbf{HPBW [$^\circ$]}      & 21       & 10       & 27       & 15       & 21       & 10       \\ \hline
        \end{tabular}
    }
    \label{tab:beamforming_metrics}
\end{table}

\subsection{Beamforming Performance Analysis}
\label{sec:BeamformingAnalysis}

This subsection evaluates the beamforming performance of the proposed BS-integrated BD-RIS under various system configurations, serving as a foundation for understanding the system’s behavior from other perspectives. As stated in (\ref{eq:signal model}), $\boldsymbol{\zeta} = \boldsymbol{\Omega} \boldsymbol{g}$ represents the effective beamforming vector at the RIS’s transmit terminal. Here, the matrix $\boldsymbol{\Omega}$ not only defines the RIS type (D-RIS or BD-RIS) but also specifies the grouping strategy employed in the BD-RIS configuration. To simplify the analysis, we consider a directional beam focused at $\phi = \theta = 0$ as the beamforming vector for configuring the BD-RIS/D-RIS, defined as $\boldsymbol{b} = \boldsymbol{a}(0,0)$.

\begin{figure}
    \centering
    \includegraphics[width=\columnwidth]{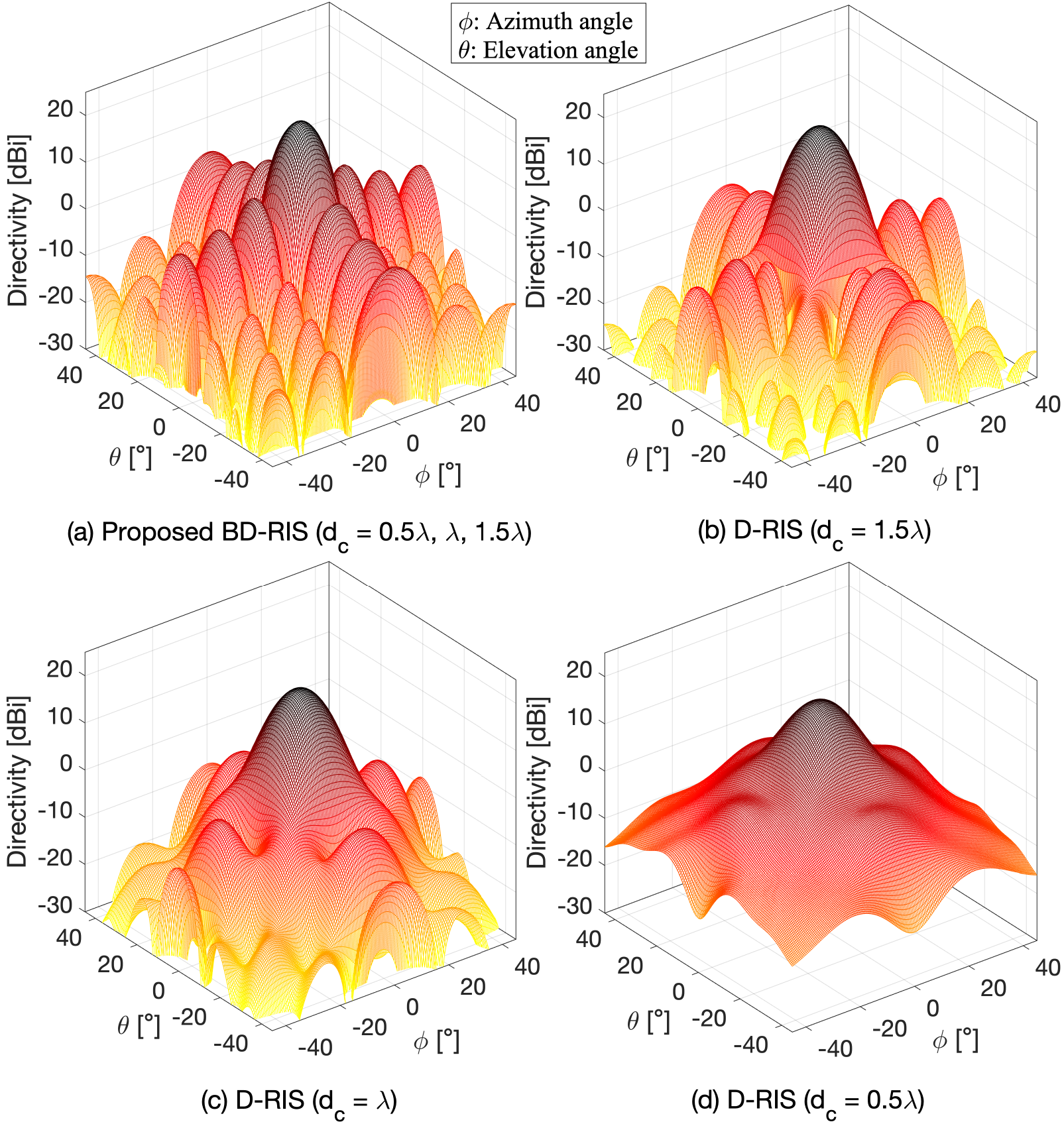}
    \caption{Beamforming performance comparison between BS-integrated BD-RIS and D-RIS for varying BS-RIS separations ($d_c$): (a) BD-RIS with $d_c = 0.5\lambda, \lambda, 1.5\lambda$, \textcolor{black}{indicating the robustness of BD-RIS beamforming gain with changing BS-RIS separation, while} (b) D-RIS with $d_c = 1.5\lambda$, (c) D-RIS with $d_c = \lambda$, and (d) D-RIS with $d_c = 0.5\lambda$\textcolor{black}{, demonstrating that placing D-RIS closer to the active antenna leads to degraded beamforming performance.}}
    \label{fig:Beamforming_dc}
\end{figure}

Fig. \ref{fig:Beamforming_ArraySize} presents the beamforming gain of the proposed BS-integrated BD-RIS compared with the benchmarks, i.e., the D-RIS and the active analog beamforming antenna array.
Overall, increasing the array size results in sharper beams with higher directivity across all structures. 
However, the D-RIS exhibits lower directivity and wider beams compared to its counterparts, while the BD-RIS achieves a beamforming performance comparable to that of the active analog beamforming antenna array. This superior performance of the BD-RIS is attributed to its fully-connected structure, which enables it to effectively compensate for the CAV.
The numerical metrics, including peak point directivity (PPD), half-power point directivity (HPPD), and half-power beamwidth (HPBW), extracted from Fig. \ref{fig:Beamforming_ArraySize} for various BS structures, are summarized in Table \ref{tab:beamforming_metrics}.

Fig. \ref{fig:Beamforming_dc} illustrates the beamforming gain of the BS-integrated BD-RIS and D-RIS for varying BS-RIS separations ($d_c$). Thanks to its ability to adjust both phase and amplitude, the BD-RIS maintains a consistent beam pattern regardless of $d_c$, while the D-RIS demonstrates reduced directivity and increased HPBW as $d_c$ decreases.
As shown in Fig. \ref{fig:CAV}, the CAV increases as  $d_c$  decreases, posing a challenge for D-RIS since it cannot compensate for the elevated CAV. 
The beamforming metrics for these configurations are summarized in Table \ref{tab:beamforming_metrics_distances}.
It is crucial to highlight that as $d_c$ increases, the multiplicative path loss becomes more pronounced, which detrimentally affects the overall system performance. Therefore, ensuring a high beamforming gain while minimizing $d_c$ is essential for optimizing performance, as elaborated in Section \ref{Sec:Simulation_AR}.

\begin{table}[!b]
\centering
\caption{Beamforming Performance Metrics for Different BS-RIS Separations}
\label{tab:beamforming_metrics_distances}
\begin{tabular}{|c|c|c|c|c|}
\hline
\multirow{2}{*}{} & \multicolumn{3}{c|}{\textbf{D-RIS}} & \textbf{BD-RIS} \\ \cline{2-5} 
                  & $\boldsymbol{d_c = 0.5\lambda}$ & $\boldsymbol{d_c = \lambda}$ & $\boldsymbol{d_c = 1.5\lambda}$ & \textbf{All} $\boldsymbol{d_c}$ \\ \hline
\textbf{PPD [dBi]} & 17.602 & 20.012 & 20.868 & 21.548 \\ \hline
\textbf{HPPD [dBi]} & 14.645 & 17.004 & 17.870 & 18.669 \\ \hline
\textbf{HPBW (\(^{\circ}\))} & 15 & 13 & 12 & 10 \\ \hline
\end{tabular}%
\end{table}



\begin{figure}
    \centering
    \includegraphics[width=\columnwidth]{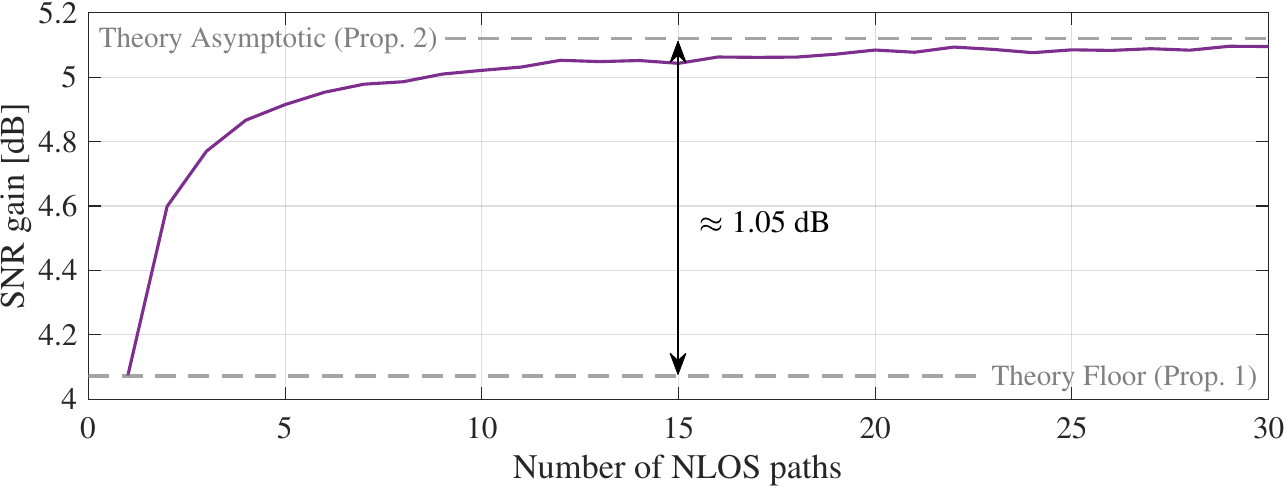}
    \caption{\color{black}SNR gain of BD-RIS over D-RIS for different numbers of NLOS clusters, illustrating the transition from sparse to rich scattering conditions. A single path per cluster is assumed in all NLOS scenarios ($L = 1$), and the RIS is configured according to Case $1$ described in Section \ref{Sec: BD-RIS Config}. \textcolor{black}{The simulation results are tightly bracketed by the theoretical bounds: the guaranteed floor in Proposition~\ref{prop:gainfloor} and the asymptotic limit in Proposition~\ref{prop:asymptotic}.}}
    \label{fig:SNR_gain}
\end{figure}

\subsection{\color{black}SNR Performance Analysis}

Fig. \ref{fig:SNR_gain} illustrates the SNR gain of BD-RIS over D-RIS, defined as $10 \log_{10}\left({\text{SNR}_{\text{BD-RIS}}}/{\text{SNR}_{\text{D-RIS}}}\right)$, across varying scattering conditions, ranging from a sparse-scattering regime (with a single NLOS path) to a rich-scattering environment \textcolor{black}{(approximated by Rayleigh fading)}. As shown, the SNR gain is at its minimum under sparse-scattering conditions; however, as the scattering environment becomes richer, BD-RIS demonstrates increasing performance gains over D-RIS. 
\textcolor{black}{These results tightly follow the theoretical bounds derived in Proposition~\ref{prop:gainfloor} and Proposition~\ref{prop:asymptotic}, which establish a gain floor and an asymptotic ceiling, respectively. For the assumed parameters in Table \ref{tab:SimulationParameters}, the CAV is approximately $1.25$ (also reflected in Fig. \ref{fig:Geneal_Achievable_Rate_Analysis}), yielding a theoretical floor of $G_{\min} \approx 4.07$ dB and an asymptotic ceiling of $G_{\max} \approx 5.12$ dB. The simulation curve in Fig. \ref{fig:SNR_gain} aligns closely with these bounds, confirming both the guaranteed minimum gain due to CAV compensation in the BS-RIS link and the additional fixed $\sim 1.05$ dB improvement under rich-scattering conditions.}

\subsection{ABER Analysis}

This subsection analyzes the ABER performance of the proposed BS-integrated BD-RIS across different LOS availability scenarios and three distinct precoding schemes. The results are benchmarked against D-RIS and active analog beamforming antenna arrays. Furthermore, as illustrated in Fig. \ref{fig:ABER}, the theoretical bound derived in (\ref{eq:Theoretical Bound}) provides a tight upper limit, validating the accuracy of our ABER simulations.

\begin{figure}
    \centering
    \includegraphics[width=\columnwidth]{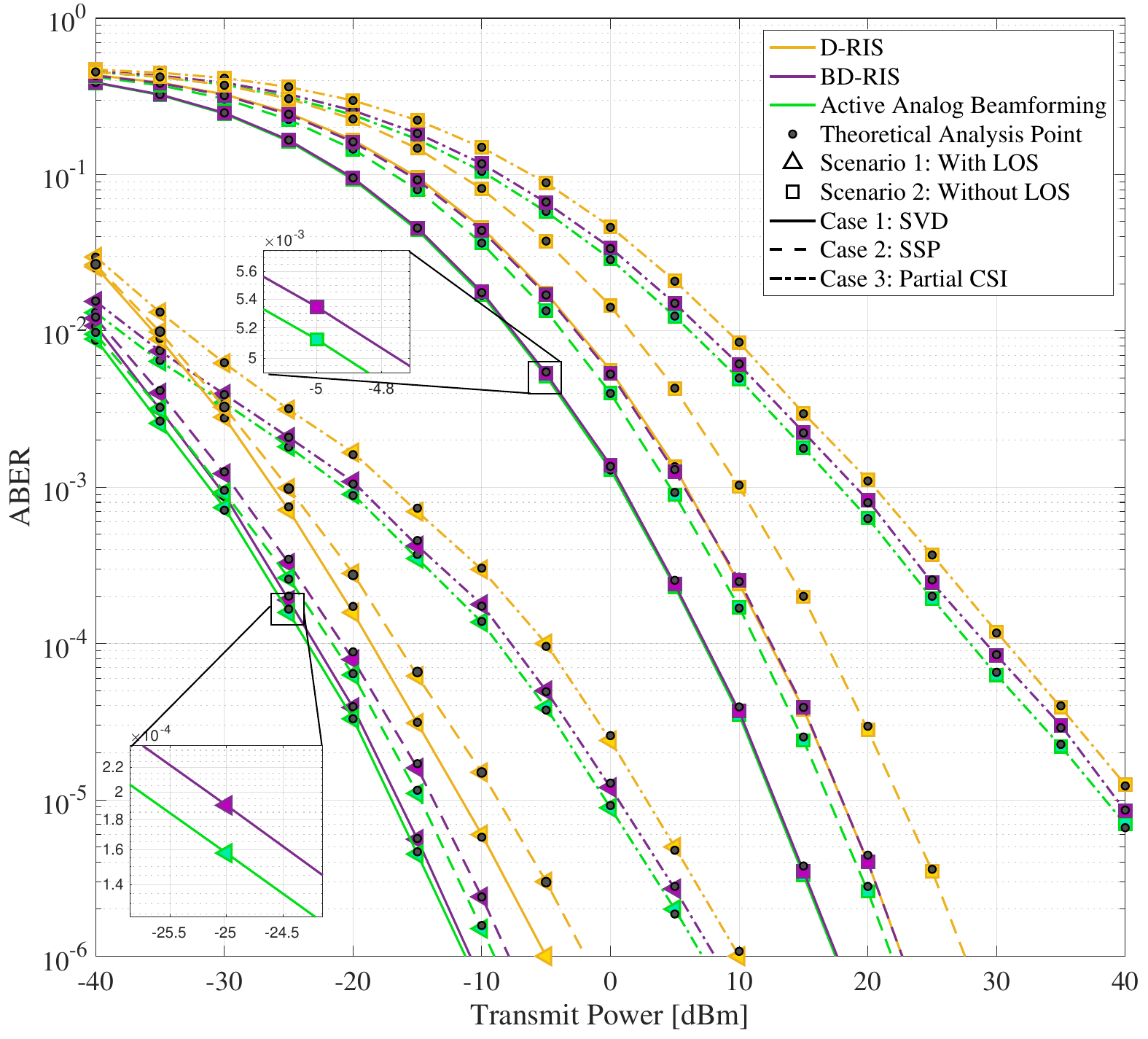}
    \caption{ABER performance of the BS-integrated BD-RIS compared to D-RIS and active analog beamforming antenna array under varying LOS conditions and precoding schemes, with theoretical validation.}
    \label{fig:ABER}
\end{figure}

As shown in Fig. \ref{fig:ABER}, the BD-RIS achieves performance close to that of the active analog beamforming antenna array across all evaluated scenarios and cases, demonstrating its robustness in delivering efficient beamforming gain under varying conditions. 
As outlined in Algorithm \ref{Alg:BDRIS_Configuration}, the beamforming vector $\boldsymbol{b}$ is employed in the BD-RIS configuration as a representation of the channel $\boldsymbol{h}$. Notably, in Case 1, $\boldsymbol{b}$ incorporates both the amplitude and phase variations. However, only the BD-RIS can fully utilize these amplitude and phase variations, whereas the benchmarks, including the active analog beamforming antenna array, are restricted to exploiting only the phase components of $\boldsymbol{b}$. This capability provides the BD-RIS with one more advantage in Case 1, resulting in ABER performance that is closer to that of the active analog beamforming antenna array compared to other cases.
In Scenario 2, where the channel $\boldsymbol{h}$ is entirely composed of NLOS components, the amplitude variations across the RIS elements are more pronounced. This further improves the ABER performance of the BD-RIS in Case 1, narrowing the gap with the active analog beamforming antenna array even more compared to Scenario 1.
In other cases, the performance gap between the BD-RIS and the active analog beamforming antenna array is solely attributed to the multiplicative path loss caused by the physical distance between the active antenna and the RIS elements.

On the other hand, the D-RIS, constrained to phase-only manipulation of the impinging signal, exhibits greater ABER performance degradation across various scenarios and cases due to both multiplicative path loss and reduced beamforming gain. As illustrated in Fig. \ref{fig:Beamforming_ArraySize} and summarized in Table \ref{tab:beamforming_metrics}, the D-RIS generates wider beams, leading to increased signal leakage into unintended paths. Consequently, as a portion of the signal experiences higher fading in these unintended paths, the ABER increases for the D-RIS configuration.

These results highlight the limitations of D-RIS and underscore the significant advantages of BD-RIS in achieving efficient beamforming gain and improved ABER performance.

\subsection{Achievable Rate Analysis}\label{Sec:Simulation_AR}


This subsection provides a detailed analysis of the achievable rate performance as a function of transmit power, array size, and BS-RIS separation ($d_c$). The achievable rate is calculated as follows:
\begin{equation}
    R = \mathbb{E}_{\boldsymbol{h}} \left[ \log_2 \left( 1 + \frac{|\sqrt{P} \boldsymbol{h}^{\mathsf{T}} \boldsymbol{\zeta} x|^2}{\sigma_n^2} \right) \right],
\end{equation}
where $x \sim \mathcal{CN}(0,1)$ is the transmitted symbol following a complex normal distribution with unit variance, allowing the analysis to be generalized for any potential constellation.


\begin{figure}
    \centering
    \includegraphics[width = \columnwidth]{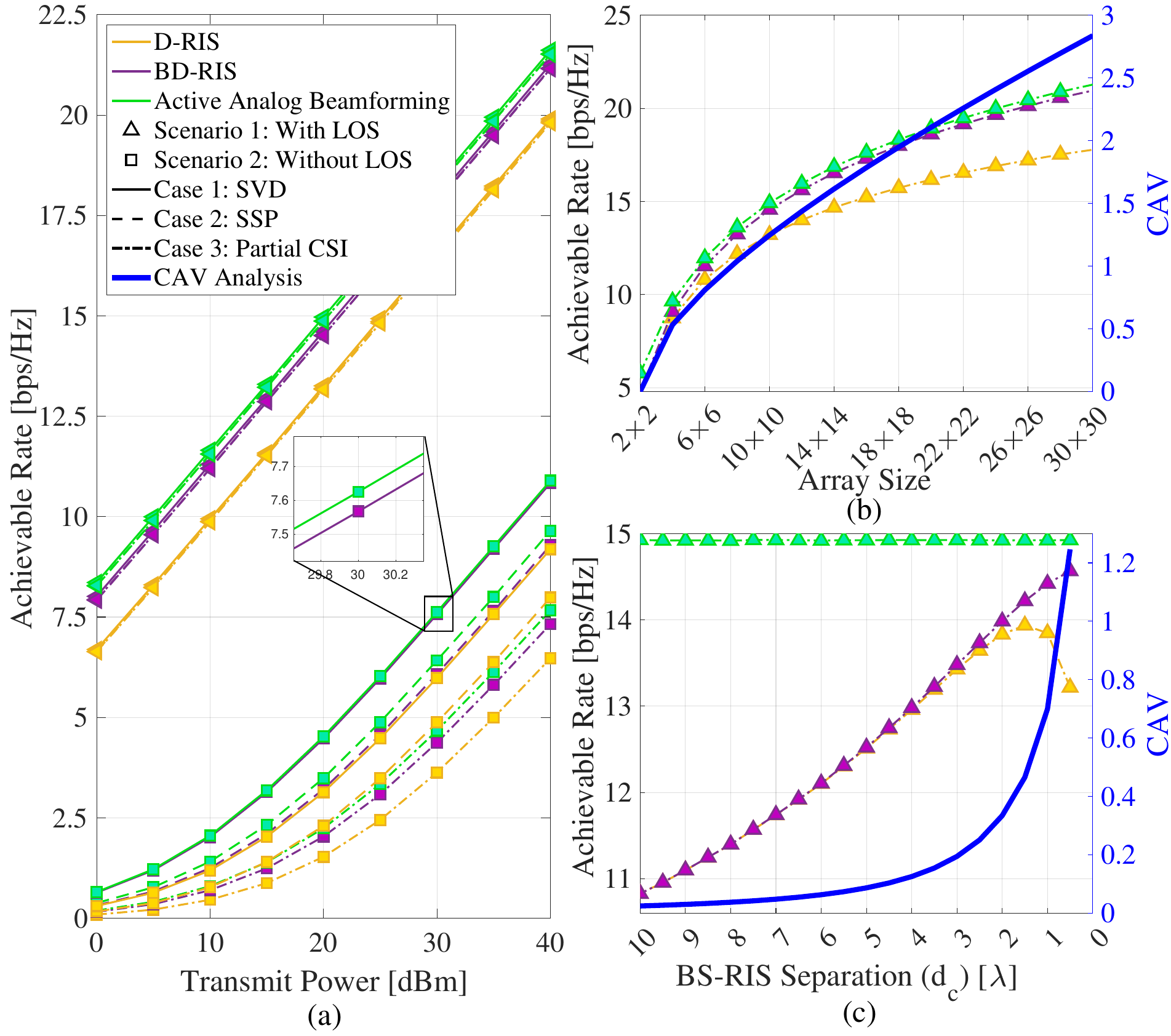}
    \caption{Achievable rate analysis \textcolor{black}{and its interaction with CAV} for the proposed BS-integrated BD-RIS and benchmarks:  
    (a) \textcolor{black}{achievable rate analysis for all considered scenarios and precoding strategies over varying transmit power levels},  
    (b) \textcolor{black}{achievable rate analysis and its interaction with CAV under Scenario 1 (LOS) and Case 3 (partial CSI) for different array sizes},  
    (c) \textcolor{black}{achievable rate analysis and its interaction with CAV under Scenario 1 (LOS) and Case 3 (partial CSI) for varying BS-RIS separation ($d_c$)}.}
    \label{fig:Geneal_Achievable_Rate_Analysis}
\end{figure}

Fig. \ref{fig:Geneal_Achievable_Rate_Analysis}(a) illustrates the achievable rate as a function of transmit power across various scenarios and precoding cases. In all simulations, the BD-RIS demonstrates performance close to that of the active analog beamforming antenna array. Notably, in Scenario 2, where LOS is unavailable, and under Case 1, the BD-RIS achieves an achievable rate comparable to the active analog beamforming antenna array. As discussed in the previous subsection, the BD-RIS’s ability to manipulate both the phase and amplitude enables it to leverage the optimal phase and amplitude provided by $\boldsymbol{v}_1$. 
In contrast, the active analog beamforming antenna array is limited to utilizing only the phase information from the dominant eigenmode $\boldsymbol{v}_1$. 
This additional capability of the BD-RIS enhances its beamforming performance, allowing it to compensate for the multiplicative path loss and achieve comparable performance to the active analog beamforming antenna array in Scenario 2. 
Similar to the ABER performance, the reduced beamforming gain of the D-RIS leads to degraded achievable rate performance across all simulations.

It is worth noting that in the presence of a dominant LOS, as in Scenario 1, the achievable rate is identical across all precoding cases, as illustrated in Fig. \ref{fig:Geneal_Achievable_Rate_Analysis}(a). For the remainder of this paper, we focus on simulations conducted under Scenario 1 using Case 3, the less complex beamforming method.
Fig. \ref{fig:Geneal_Achievable_Rate_Analysis}(b) presents the achievable rate (scaled on left $y$-axis) as a function of the array size, alongside the corresponding CAV (blue curve scaled on right $y$-axis), comparing the performance of the proposed BS-integrated BD-RIS against benchmark schemes. For a small RIS array of size $2 \times 2$, the BD-RIS and D-RIS exhibit identical performance. This is because, in this configuration, the distances between the active antenna and all RIS elements ($d_m$) are equal, resulting in uniform channel gain amplitudes across the elements; hence $\textrm{CAV} = 0$ and there is no need for amplitude compensation. However, as the array size increases, the channel gain amplitudes across the RIS elements start to vary, as explained in Section \ref{sec:CAV}.
As shown in Fig. \ref{fig:Geneal_Achievable_Rate_Analysis}(b), larger array sizes lead to increased CAV, introducing greater discrepancies in channel gain amplitudes across RIS elements and emphasizing the need to compensate for these variations. 
Since the D-RIS cannot address amplitude variations, the performance gap between the D-RIS and the active analog beamforming antenna array widens with increasing array size. 
In contrast, the BD-RIS, with its ability to compensate for both amplitude and phase variations, maintains a constant performance gap compared to the active analog beamforming antenna array as the array size increases. As previously discussed, this constant gap is attributed to the multiplicative path loss.


The BS-RIS separation ($d_c$) is another critical factor influencing the CAV, as outlined in Section \ref{sec:CAV}. To provide better insight, the CAV is plotted on the right  $y$-axis in Fig. \ref{fig:Geneal_Achievable_Rate_Analysis}(c), alongside the achievable rate (left  $y$-axis) for varying BS-RIS separation distances. As shown in Fig. \ref{fig:Geneal_Achievable_Rate_Analysis}(c), at sufficiently large separations, a reduction in $d_c$ generally leads to an increase in the achievable rate due to reduced multiplicative path loss. However, as $d_c$ decreases further, the D-RIS begins to exhibit noticeable performance degradation compared to the BD-RIS. 
This behavior arises because the increasing CAV at smaller $d_c$ values is more pronounced and demands compensation to maintain effective beamforming gain. While the BD-RIS, with its inter-element connections, can effectively mitigate amplitude variations, the D-RIS lacks this capability.
As shown in Fig. \ref{fig:Geneal_Achievable_Rate_Analysis}(c), when $d_c$ falls below a certain threshold ($d_c < 1.5 \lambda$ in the configuration considered here), the achievable rate performance of the D-RIS starts to decline with further decreasing $d_c$. 
According to Fig. \ref{fig:Beamforming_dc}, the beamforming performance of the D-RIS deteriorates with decreasing $d_c$, resulting in wider beams and reduced directivity. This suboptimal beamforming performance for $d_c < 1.5 \lambda$ leads to a decrease in the achievable rate for the D-RIS, as the negative impact of poor beamforming gain outweighs the positive effect of reduced multiplicative path loss. In contrast, the BD-RIS consistently demonstrates enhanced performance as $d_c$ decreases, attributed to its ability to effectively compensate for significant CAV and sustain robust beamforming. This enables the BD-RIS to fully capitalize on the reduced multiplicative path loss associated with smaller $d_c$.

\begin{figure}
    \centering
    \includegraphics[width=\columnwidth]{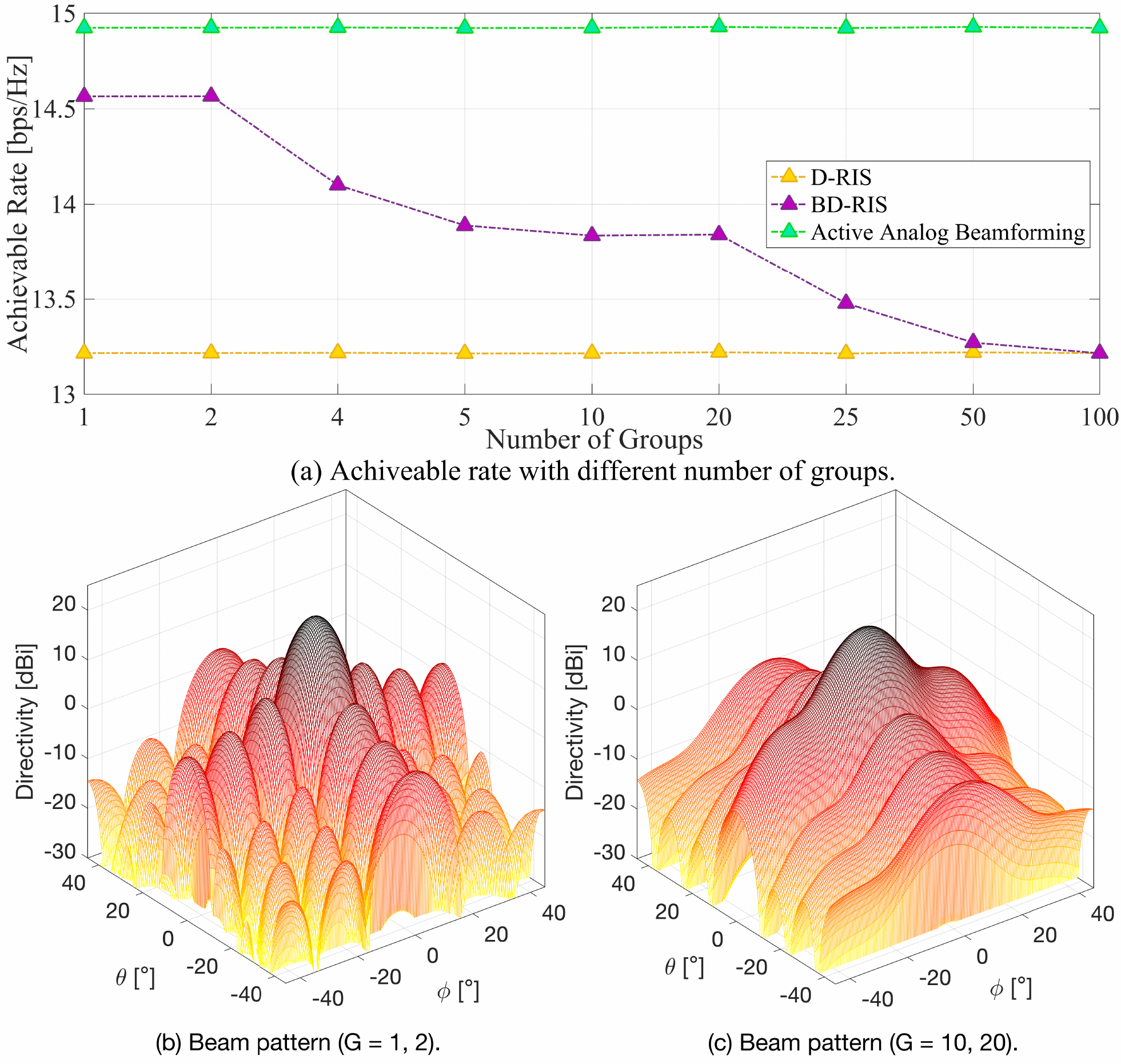}
    \caption{Achievable rate and beam pattern comparison of the BS-integrated BD-RIS under different group configurations: (a) Achievable rate performance for varying numbers of groups; (b) beam pattern comparison for fully-connected BD-RIS ($G = 1$) and group-connected BD-RIS ($G = 2$); (c) beam pattern comparison for group-connected BD-RIS configurations with $G = 10$ and $G = 20$.}
    \label{fig:BeamPattern_GroupConnection}
\end{figure}

\subsection{Group-Connected BD-RIS Structure}

In this subsection, we examine the impact of the group-connected BD-RIS, which features a less complex circuit architecture compared to the fully-connected structure, on the achievable rate performance and beamforming gain of the proposed system.
Fig. \ref{fig:BeamPattern_GroupConnection}(a) illustrates the achievable rate performance for different numbers of groups. As the number of groups increases, the inter-element connections within each group decrease, as each fully-connected group comprises fewer elements. 
\textcolor{black}{While reducing inter-element connections leads to lower circuit complexity, it also limits the available degrees of freedom for effectively compensating CAV, which in turn degrades the achievable rate. Nonetheless, specific cases illustrated in Fig. \ref{fig:BeamPattern_GroupConnection}(a) deviate from this general trend. In particular, the performance remains almost unchanged between $G = 1$ and $G = 2$, as well as between $G = 10$ and $G = 20$.}

\begin{figure}
    \centering
    \includegraphics[width=\columnwidth]{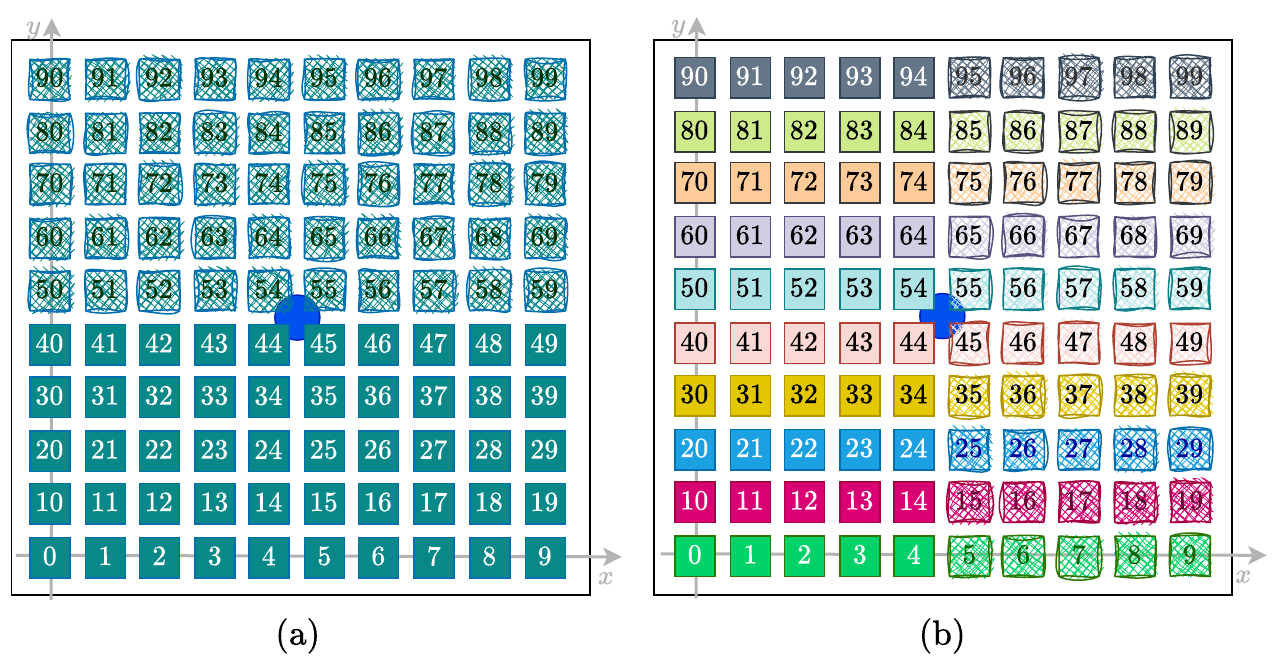}
    \caption{Grouping configurations of the BD-RIS: (a) a fully-connected structure transitioning to a $2$-group-connected structure ($G = 2$) without performance loss; (b) a $10$-group-connected structure transitioning to a $20$-group-connected structure ($G = 20$) while maintaining performance integrity.}
    \label{fig:GroupConnectedSymmetry}
\end{figure}

The reason for this behavior is illustrated in the beam patterns in Figs. \ref{fig:BeamPattern_GroupConnection}(b) and (c). For $G = 1$ and $G = 2$ (Fig. \ref{fig:BeamPattern_GroupConnection}(b)), the beam patterns are identical, and a similar observation holds for $G = 10$ and $G = 20$ (Fig. \ref{fig:BeamPattern_GroupConnection}(c)). This phenomenon can be attributed to the geometric symmetry in the group configurations, as shown in Fig. \ref{fig:GroupConnectedSymmetry}. The blue circle in Fig. \ref{fig:GroupConnectedSymmetry} denotes the location of the active antenna array, which is positioned at the back of the RIS at a distance $d_c$ from its center.
In Fig. \ref{fig:GroupConnectedSymmetry}(a), a fully-connected structure is depicted, capable of compensating for CAV across the entire BD-RIS. However, this fully-connected structure can be divided into two groups for a $2$-group-connected configuration, where one group is highlighted with solid colors, and the other is represented with cross-hatches. 
These groups exhibit geometric symmetry around the active antenna. Hence, it is evident that the set of $[\boldsymbol{g}]_m$ values for each group is identical, leading to an equivalent CAV value for both groups, \textcolor{black}{as proven in Proposition~\ref{prop:SymmetricGeometryCAV}}. Furthermore, \textcolor{black}{based on this proposition}, when these two sets are combined, the CAV of the entire BD-RIS remains equal to the CAV of each individual group. Consequently, a 2-group-connected BD-RIS structure with linear permutation suffices to effectively compensate for CAV across the entire BD-RIS.\footnote{Clearly, the 2-group-connected structure can be further divided into a 4-group-connected structure while maintaining identical CAV values due to the symmetry around the active antenna. However, achieving this requires more advanced permutations beyond the linear approach, which will be explored in our future work.} 
Similarly, Fig. \ref{fig:GroupConnectedSymmetry}(b) illustrates a $10$-group-connected structure, where the elements in each row are assigned to the same group. This structure can be further divided into two symmetric sub-groups around the active antenna, ensuring that the CAV values for each sub-group remain identical and equivalent to the original group. Consequently, inter-element connections can be minimized without compromising performance in such symmetric configurations.

It is worth noting that, as depicted in Fig. \ref{fig:GroupConnectedSymmetry}(b), the $10$-group-connected structure incorporates inter-element connections exclusively along one dimension (the $x$-axis), with no connections along the $y$-axis. This configuration leads to a narrow HPBW in the elevation dimension while maintaining a broader beam in the azimuth dimension, as demonstrated in Fig. \ref{fig:BeamPattern_GroupConnection}(b).


\section{Conclusion}\label{Sec:Conclusion}
This paper studied a novel array architecture that enables passive beamforming at the BS, achieving beamforming gains comparable to those of active analog beamforming antenna arrays. By integrating the BD-RIS into the BS and leveraging its superior passive beamforming capability, a SISO system can emulate MISO performance.
Our analysis demonstrated that the beamforming gain of the BS-integrated BD-RIS is robust to variations in RIS dimensions and BS-RIS separation, owing to its inter-element connection structure. In contrast, the performance of traditional D-RIS is highly sensitive to these critical design parameters. Consequently, the only notable performance degradation in the BS-integrated BD-RIS arises from the multiplicative path loss, which can be minimized by reducing the BS-RIS separation without concerns about beamforming degradation at short separations.
To address circuit complexity, we proposed a simple grouping strategy that reduces the number of inter-element connections while maintaining equivalent beamforming performance by grouping the elements symmetrically around the active antenna.
Comprehensive simulation results validated the efficiency of the proposed BS-integrated BD-RIS architecture and demonstrated its superiority over traditional D-RIS systems. Furthermore, it was shown that the proposed architecture achieves performance levels close to those of active analog beamforming antenna arrays, offering an affordable alternative in terms of cost and power consumption. \textcolor{black}{Future work will focus on extending the proposed system model beyond the unilateral approximation by incorporating EM interactions between the active antenna and RIS elements.}

\appendices

{\color{black}
\section{Proof of Proposition~\ref{prop:gainfloor}}\label{app:gainfloor}
Consider the limiting case where the RIS-UE channel exhibits identical amplitudes across all elements. This situation naturally arises in a single-path channel, where the impinging wavefront preserves constant magnitude across the RIS elements (see \cite{9514409}).
Hence, $|[\boldsymbol{h}]_m|=\beta$, where $\beta$ is a constant equal to the path-gain amplitude. Substituting this into (\ref{eq:SNR D-RIS}) and (\ref{eq: SNR BD-RIS}), the SNR expressions for D-RIS and BD-RIS become proportional to $\beta^2 \left(\sum_{m=1}^M |[\boldsymbol{g}]_m|\right)^2$ and $M \beta^2 \sum_{m=1}^{M}|[\boldsymbol{g}]_m|^2$, respectively;
thus, exploiting (\ref{eq:mean (CAV)}), the SNR ratio is given by
\begin{equation}\label{eq:SNR ratio}
    \frac{\mathrm{SNR}_{\mathrm{BD-RIS}}}{\mathrm{SNR}_{\mathrm{D-RIS}}} = \frac{\frac{1}{M} \sum_{m=1}^{M} |[\boldsymbol{g}]_m|^2}{\left(\frac{1}{M}\sum_{m=1}^M |[\boldsymbol{g}]_m|\right)^2} = \frac{\frac{1}{M} \sum_{m=1}^{M} |[\boldsymbol{g}]_m|^2}{\mu_{|\boldsymbol{g}|}^2}.
\end{equation}
From (\ref{eq:std (CAV)}) and (\ref{eq:mean (CAV)}), we have
\begin{equation}\label{eq:BS-RIS channel norm}
    \frac{1}{M} \sum_{m=1}^{M} |[\boldsymbol{g}]_m|^2 = \sigma_{|\boldsymbol{g}|}^2 + \mu_{|\boldsymbol{g}|}^2.
\end{equation}
Substituting (\ref{eq:BS-RIS channel norm}) into (\ref{eq:SNR ratio}) yields
\begin{equation}
    \frac{\mathrm{SNR}_{\mathrm{BD-RIS}}}{\mathrm{SNR}_{\mathrm{D-RIS}}} = 1 + \frac{\sigma_{|\boldsymbol{g}|}^2}{\mu_{|\boldsymbol{g}|}^2} = 1 + \mathrm{CAV}^2.
\end{equation}
Converting to dB, the minimum guaranteed SNR gain is $G_{\min} \;=\; 10\log_{10}\!\big(1+\mathrm{CAV}^2\big)$, 
which completes the proof of Proposition~\ref{prop:gainfloor}.

\section{Proof of Proposition~\ref{prop:asymptotic}}\label{app:asymptotic}

For the Rayleigh fading case, each RIS-UE channel entry follows a complex Gaussian distribution, i.e., $[\boldsymbol{h}]_m \sim \mathcal{CN} (0, \sigma_h^2)$.
Consequently, $|[\boldsymbol{h}]_m|$ is Rayleigh distributed with mean, second moment, and variance given by
\begin{equation}\label{eq:mean Rayleigh}
    \mathbb{E}[|[\boldsymbol{h}]_m|] = \frac{\sqrt{\pi}}{2} \sigma_h,
\end{equation}
\begin{equation}\label{eq:second momentum Rayleigh}
    \mathbb{E}[|[\boldsymbol{h}]_m|^2] = \sigma_h^2,
\end{equation}
\begin{equation}\label{eq:variance Rayleigh}
    \mathbb{V}[|[\boldsymbol{h}]_m|] = (1 - \frac{\pi}{4}) \sigma_h^2.
\end{equation}
By taking the expectation of both sides of (\ref{eq:SNR D-RIS}), we obtain $\mathbb{E}[\mathrm{SNR}_{\mathrm{D-RIS}}] \propto \mathbb{E} \left[ \left( \sum_{m=1}^M |[\boldsymbol{h}]_m | |[\boldsymbol{g}]_m| \right)^2 \right]$.
To evaluate this, we use the identity 
\begin{equation}\label{eq:variance-mean-momentum}
    \mathbb{E}[X^2] = (\mathbb{E}[X])^2 + \mathbb{V}[X],
\end{equation} 
where $X = \sum_{m=1}^M |[\boldsymbol{h}]_m | |[\boldsymbol{g}]_m|$. The mean term simplifies as
\begin{equation}
    \mathbb{E}[X] = \mathbb{E}\left[ \sum_{m=1}^M |[\boldsymbol{h}]_m | |[\boldsymbol{g}]_m| \right] = \sum_{m=1}^M \mathbb{E} \left[ |[\boldsymbol{h}]_m | |[\boldsymbol{g}]_m| \right].
\end{equation}
Because the BS-RIS channel $\boldsymbol{g}$ is deterministic, it follows that
\begin{equation}\label{eq:mean X}
    \mathbb{E}[X] = \sum_{m=1}^M |[\boldsymbol{g}]_m| \mathbb{E} \left[ |[\boldsymbol{h}]_m | \right].
\end{equation}
Substituting (\ref{eq:mean (CAV)}) together with (\ref{eq:mean Rayleigh}) into (\ref{eq:mean X}), we get
\begin{equation}\label{eq:Final mean X}
    \mathbb{E}[X] = M \frac{\sqrt{\pi}}{2} \sigma_h \mu_{|\boldsymbol{g}|}.
\end{equation}
Similarly, with independent channel coefficients $[\boldsymbol{h}]_m$ and using (\ref{eq:BS-RIS channel norm}), the variance term can be expressed as
\begin{equation}\label{eq:Final variance X}
    \mathbb{V}[X] = \sum_{m=1}^{M} |[\boldsymbol{g}]_m|^2 \mathbb{V}\left[ |[\boldsymbol{h}]_m | \right] = M \sigma_h^2 (\sigma_{|\boldsymbol{g}|}^2 + \mu_{|\boldsymbol{g}|}^2) (1 - \frac{\pi}{4}).
\end{equation}
By substituting (\ref{eq:Final mean X}) and (\ref{eq:Final variance X}) into (\ref{eq:variance-mean-momentum}), we obtain
\begin{equation}
    \begin{aligned}
        \mathbb{E} & [\mathrm{SNR}_{\mathrm{D-RIS}}] \propto \mathbb{E}\left[X^2\right] \\
        & = \frac{\pi}{4}M^2 \sigma_h^2 \mu_{|\boldsymbol{g}|}^2 + M \sigma_h^2 \left( \sigma_{|\boldsymbol{g}|}^2 + \mu_{|\boldsymbol{g}|}^2 \right) \left( 1 - \frac{\pi}{4} \right).
    \end{aligned}
\end{equation}
Similarly, taking expectation of (\ref{eq: SNR BD-RIS}) yields
\begin{equation}\label{eq:mean BD-RIS}
    \mathbb{E} \left[ \mathrm{SNR}_{\mathrm{BD-RIS}} \right] \propto \mathbb{E} \left[ \| \boldsymbol{h} \|^2 \| \boldsymbol{g} \|^2 \right] = \mathbb{E}[\| \boldsymbol{h}\|^2] \sum_{m=1}^M |[\boldsymbol{g}]_m|^2.
\end{equation}
From (\ref{eq:second momentum Rayleigh}), it follows that
\begin{equation}\label{eq:mean h}
    \mathbb{E}[\| \boldsymbol{h} \|^2] = \sum_{m=1}^M \mathbb{E}[|[\boldsymbol{h}]_m|^2] = M \sigma_h^2.
\end{equation}
Substituting (\ref{eq:BS-RIS channel norm}) and (\ref{eq:mean h}) into (\ref{eq:mean BD-RIS}) gives
\begin{equation}
    \mathbb{E} \left[ \mathrm{SNR}_{\mathrm{BD-RIS}} \right] \propto M^2 \sigma_h^2 (\sigma_{|\boldsymbol{g}|}^2 + \mu_{|\boldsymbol{g}|}^2).
\end{equation}
The SNR ratio between BD-RIS and D-RIS then simplifies to
\begin{equation}
    \frac{\mathbb{E} \left[ \mathrm{SNR}_{\mathrm{BD-RIS}} \right]}{\mathbb{E} \left[ \mathrm{SNR}_{\mathrm{D-RIS}} \right]} = \frac{M^2 (\sigma_{|\boldsymbol{g}|}^2 + \mu_{|\boldsymbol{g}|}^2)}{\frac{\pi}{4}M^2  \mu_{|\boldsymbol{g}|}^2 + M \left( \sigma_{|\boldsymbol{g}|}^2 + \mu_{|\boldsymbol{g}|}^2 \right) \left( 1 - \frac{\pi}{4} \right)}.
\end{equation}
For sufficiently large RIS sizes ($M\gg 1$), retaining only the leading-order contribution in $M$ yields:
\begin{equation}
    \frac{\mathbb{E} \left[ \mathrm{SNR}_{\mathrm{BD-RIS}} \right]}{\mathbb{E} \left[ \mathrm{SNR}_{\mathrm{D-RIS}} \right]} \approx \frac{\sigma_{|\boldsymbol{g}|}^2 + \mu_{|\boldsymbol{g}|}^2}{\frac{\pi}{4} \mu_{|\boldsymbol{g}|}^2} = \frac{4}{\pi} \left( 1 + \mathrm{CAV}^2 \right).
\end{equation}
Thus, the asymptotic gain in dB is $G_{\max} = G_{\min} + 10\log_{10}\left(\frac{4}{\pi}\right)$.
}

{\color{black}
\section{Proof of Proposition \ref{prop:SymmetricGeometryCAV}} \label{app:SymmetricGeometryCAV}

   We first consider the case of two subsets and later generalize the result to the $K$-subset case.  Let $A = A_1 \cup A_2$ with $A_1 \cap A_2 = \emptyset$, where each subset contains $N$ elements. From (\ref{eq:std (CAV)}) and (\ref{eq:mean (CAV)}), it is straightforward to show that the sample mean and variance of each subset are identical, i.e., $\mu_{A_1} = \mu_{A_2}$ and 
    $\sigma_{A_1} = \sigma_{A_2}$. Hence, $\mathrm{CAV}(A_1) = \mathrm{CAV}(A_2)$.

   For simplicity of notation, we define $A_j = \{ x_1, \dots, x_N \}$, where $x_i = |[\boldsymbol{g}]_i|$. For the union set $A$, which contains $2N$ elements, the sample mean is
    \begin{equation}
        \begin{aligned}
            \mu_A &= \frac{1}{2N} \sum_{x \in A} x
            = \frac{1}{2N} \left( \sum_{x \in A_1} x + \sum_{x \in A_2} x \right) \\
            &= \frac{1}{2N} \left( N \mu_{A_1} + N \mu_{A_2} \right) = \mu_{A_1}.
        \end{aligned}
    \end{equation}    
    Similarly, the sample variance is 
    \begin{equation}
    \begin{aligned}
    \sigma_A^2 &= \frac{1}{2N} \sum_{x \in A} (x - \mu_A)^2 \\
            &= \frac{1}{2N} \left( \sum_{x \in A_1} (x - \mu_{A_1})^2 
            + \sum_{x \in A_2} (x - \mu_{A_2})^2 \right)\\
            &= \frac{1}{2N} \left( N \sigma_{A_1}^2 + N \sigma_{A_2}^2 \right) 
            = \sigma_{A_1}^2.
    \end{aligned}
    \end{equation}
    Hence, $\mu_A = \mu_{A_1}$ and $\sigma_A = \sigma_{A_1}$, which implies $\mathrm{CAV}(A) = \mathrm{CAV}(A_1) = \mathrm{CAV}(A_2)$.   
    
    \underline{\textit{Generalization:}}  
    Now suppose that $A$ is partitioned into $K$ disjoint symmetric subsets $A_1, A_2, \dots, A_K$, each of size $N$, such that all subsets have identical statistics ($\mu_{A_k} = \mu_{A_1}$, $\sigma_{A_k} = \sigma_{A_1}$ for all $k$). Then, the sample mean of $A$ is  
    \begin{equation}
    \mu_A = \frac{1}{KN} \sum_{k=1}^K \sum_{x \in A_k} x 
        = \frac{1}{KN} \sum_{k = 1}^K N\mu_{A_1} 
        = \mu_{A_1}.
    \end{equation}  
    Similarly, the sample variance is  
    \begin{equation}
        \sigma_A^2 = \frac{1}{KN} \sum_{k=1}^K \sum_{x \in A_k} (x - \mu_{A_1})^2
        = \frac{1}{KN} \sum_{k = 1}^K N\sigma_1^2
        = \sigma_{A_1}^2.
    \end{equation} 
    Hence, $\mathrm{CAV}(A) = \mathrm{CAV}(A_1) = \cdots = \mathrm{CAV}(A_K)$.
}



\ifCLASSOPTIONcaptionsoff
  \newpage
\fi




\bibliographystyle{IEEEtran}
\bibliography{references}

\begin{thebibliography}{10}
\providecommand{\url}[1]{#1}
\csname url@samestyle\endcsname
\providecommand{\newblock}{\relax}
\providecommand{\bibinfo}[2]{#2}
\providecommand{\BIBentrySTDinterwordspacing}{\spaceskip=0pt\relax}
\providecommand{\BIBentryALTinterwordstretchfactor}{4}
\providecommand{\BIBentryALTinterwordspacing}{\spaceskip=\fontdimen2\font plus
\BIBentryALTinterwordstretchfactor\fontdimen3\font minus \fontdimen4\font\relax}
\providecommand{\BIBforeignlanguage}[2]{{%
\expandafter\ifx\csname l@#1\endcsname\relax
\typeout{** WARNING: IEEEtran.bst: No hyphenation pattern has been}%
\typeout{** loaded for the language `#1'. Using the pattern for}%
\typeout{** the default language instead.}%
\else
\language=\csname l@#1\endcsname
\fi
#2}}
\providecommand{\BIBdecl}{\relax}
\BIBdecl

\bibitem{10176315}
M.~Raeisi, A.~Koc, I.~Yildirim, E.~Basar, and T.~Le-Ngoc, ``Cluster index modulation for reconfigurable intelligent surface-assisted {mmWave} massive {MIMO},'' \emph{IEEE Trans. Wireless Commun.}, vol.~23, no.~2, pp. 1581--1591, Jul. 2023.

\bibitem{wang2020joint}
P.~Wang, J.~Fang, L.~Dai, and H.~Li, ``Joint transceiver and large intelligent surface design for massive {MIMO} {mmWave} systems,'' \emph{IEEE Trans. Wireless. Commun.}, vol.~20, no.~2, pp. 1052--1064, Oct. 2020.

\bibitem{raeisi2022cluster}
M.~Raeisi, A.~Koc, E.~Basar, and T.~Le-Ngoc, ``Cluster index modulation for {mmWave} communication systems,'' \emph{Front. Comms. Net.}, Feb. 2022.

\bibitem{mahmood20223}
M.~Mahmood, A.~Koc, and T.~Le-Ngoc, ``{3-D} antenna array structures for millimeter wave multi-user massive {MIMO} hybrid precoder design: A performance comparison,'' \emph{IEEE Commun. Lett.}, vol.~26, no.~6, pp. 1393--1397, Mar. 2022.

\bibitem{8891298}
A.~Koc, A.~Masmoudi, and T.~Le-Ngoc, ``Angular-based {3D} hybrid precoding for {URA} in multi-user massive {MIMO} systems,'' in \emph{IEEE 90th Veh. Technol. Conf. (VTC2019-Fall), Honolulu, HI, USA}, Sep. 2019, pp. 1--5.

\bibitem{koc2020hybrid}
A.~Koc and T.~Le-Ngoc, ``Hybrid millimeter-wave massive {MIMO} systems with low {CSI} overhead and few-bit {DACs/ADCs},'' in \emph{IEEE 92nd Veh. Technol. Conf. (VTC2020-Fall)}, Dec. 2020, pp. 1--5.

\bibitem{yu2018hardware}
X.~Yu, J.~Zhang, and K.~B. Letaief, ``A hardware-efficient analog network structure for hybrid precoding in millimeter wave systems,'' \emph{IEEE J. Sel. Top. Signal Process.}, vol.~12, no.~2, pp. 282--297, Mar. 2018.

\bibitem{basar2019wireless}
E.~Basar, M.~Di~Renzo, J.~De~Rosny, M.~Debbah, M.-S. Alouini, and R.~Zhang, ``Wireless communications through reconfigurable intelligent surfaces,'' \emph{IEEE Access}, vol.~7, pp. 116\,753--116\,773, Aug. 2019.

\bibitem{raeisi2024comprehensive}
M.~Raeisi, A.~Khaleel, M.~C. Ilter, M.~Gerami, and E.~Basar, ``A comprehensive design framework for {UE}-side and {BS}-side {RIS} deployments,'' \emph{IEEE Wireless Commun.}, vol.~32, no.~3, pp. 148--155, Mar. 2025.

\bibitem{10515204}
J.~An \emph{et~al.}, ``Stacked intelligent metasurface-aided {MIMO} transceiver design,'' \emph{IEEE Wireless Commun.}, vol.~31, no.~4, pp. 123--131, Apr. 2024.

\bibitem{9998527}
Z.~Zhang, L.~Dai, X.~Chen, C.~Liu, F.~Yang, R.~Schober, and H.~V. Poor, ``Active {RIS} vs. passive {RIS}: {Which} will prevail in {6G}?'' \emph{IEEE Trans. Commun.}, vol.~71, pp. 1707--1725, Dec. 2022.

\bibitem{9570143}
X.~Mu, Y.~Liu, L.~Guo, J.~Lin, and R.~Schober, ``Simultaneously transmitting and reflecting ({STAR}) {RIS} aided wireless communications,'' \emph{IEEE Trans. Wireless Commun.}, vol.~21, no.~5, pp. 3083--3098, Oct. 2022.

\bibitem{li2023reconfigurable}
H.~Li, S.~Shen, M.~Nerini, and B.~Clerckx, ``Reconfigurable intelligent surfaces 2.0: Beyond diagonal phase shift matrices,'' \emph{IEEE Commun. Mag.}, vol.~62, no.~3, pp. 102--108, Nov. 2023.

\bibitem{9685418}
K.~Liu, Z.~Zhang, and L.~Dai, ``User-side {RIS}: Realizing large-scale array at user side,'' in \emph{IEEE Glob. Commun. Conf. (GLOBECOM)}, Dec. 2021, pp. 01--06.

\bibitem{9598898}
K.~Liu, Z.~Zhang, L.~Dai, and L.~Hanzo, ``Compact user-specific reconfigurable intelligent surfaces for uplink transmission,'' \emph{IEEE Trans. Commun.}, vol.~70, no.~1, pp. 680--692, Nov. 2022.

\bibitem{10144102}
Y.~Huang, L.~Zhu, and R.~Zhang, ``Integrating intelligent reflecting surface into base station: Architecture, channel model, and passive reflection design,'' \emph{IEEE Trans. Commun.}, vol.~71, no.~8, pp. 5005--5020, 2023.

\bibitem{9991837}
W.~Du, Z.~Chu, G.~Chen, P.~Xiao, Z.~Lin, C.~Huang, and W.~Hao, ``Hybrid beamforming design for {ITS}-assisted wireless networks,'' \emph{IEEE Wireless Commun. Lett.}, vol.~12, no.~3, pp. 451--455, Mar. 2023.

\bibitem{10535263}
Q.~Li, M.~El-Hajjar, C.~Xu, J.~An, C.~Yuen, and L.~Hanzo, ``Stacked intelligent metasurfaces for holographic {MIMO}-aided cell-free networks,'' \emph{IEEE Trans. Commun.}, vol.~72, no.~11, pp. 7139--7151, May 2024.

\bibitem{10279173}
J.~An, M.~Di~Renzo, M.~Debbah, and C.~Yuen, ``Stacked intelligent metasurfaces for multiuser beamforming in the wave domain,'' in \emph{IEEE Int. Conf. Commun. (ICC), Rome, Italy}, May. 28 - Jun. 01, 2023, pp. 2834--2839.

\bibitem{10679332}
Z.~Li, J.~An, and C.~Yuen, ``Stacked intelligent metasurfaces for fully-analog wideband beamforming design,'' in \emph{IEEE VTS Asia Pac. Wireless Commun. Symp. (APWCS), Singapore}, Aug. 2024.

\bibitem{10445164}
X.~Yao, J.~An, L.~Gan, M.~Di~Renzo, and C.~Yuen, ``Channel estimation for stacked intelligent metasurface-assisted wireless networks,'' \emph{IEEE Wireless Commun. Lett.}, vol.~13, no.~5, pp. 1349--1353, Feb. 2024.

\bibitem{9913356}
H.~Li, S.~Shen, and B.~Clerckx, ``Beyond diagonal reconfigurable intelligent surfaces: {From} transmitting and reflecting modes to single-, group-, and fully-connected architectures,'' \emph{IEEE Trans. Wireless Commun.}, vol.~22, no.~4, pp. 2311--2324, Oct. 2023.

\bibitem{10158988}
------, ``Beyond diagonal reconfigurable intelligent surfaces: A multi-sector mode enabling highly directional full-space wireless coverage,'' \emph{IEEE Journal on Selected Areas in Communications}, vol.~41, no.~8, pp. 2446--2460, 2023.

\bibitem{9737373}
Q.~Li, M.~El-Hajjar, I.~Hemadeh, A.~Shojaeifard, A.~A.~M. Mourad, B.~Clerckx, and L.~Hanzo, ``Reconfigurable intelligent surfaces relying on non-diagonal phase shift matrices,'' \emph{IEEE Trans. Veh. Technol.}, vol.~71, no.~6, pp. 6367--6383, Mar. 2022.

\bibitem{10472097}
Q.~Li, M.~El-Hajjar, I.~Hemadeh, A.~Shojaeifard, and L.~Hanzo, ``Coordinated reconfigurable intelligent surfaces: Non-diagonal group-connected design,'' \emph{IEEE Trans. Veh. Technol.}, vol.~73, no.~7, pp. 10\,811--10\,816, Mar. 2024.

\bibitem{10643263}
Y.~Dong, Q.~Li, S.~X. Ng, and M.~El-Hajjar, ``Reconfigurable intelligent surface relying on low-complexity joint sector non-diagonal structure,'' \emph{IEEE Open J. Veh. Technol.}, vol.~5, pp. 1106--1123, Aug. 2024.

\bibitem{9514409}
S.~Shen, B.~Clerckx, and R.~Murch, ``Modeling and architecture design of reconfigurable intelligent surfaces using scattering parameter network analysis,'' \emph{IEEE Trans. Wireless Commun.}, vol.~21, no.~2, pp. 1229--1243, Aug. 2022.

\bibitem{10308579}
A.~Mishra, Y.~Mao, C.~D’Andrea, S.~Buzzi, and B.~Clerckx, ``Transmitter side beyond-diagonal reconfigurable intelligent surface for massive {MIMO} networks,'' \emph{IEEE Wireless Commun. Lett.}, vol.~13, no.~2, pp. 352--356, Nov. 2024.

\bibitem{10693959}
K.~Chen and Y.~Mao, ``Transmitter side beyond-diagonal {RIS} for {mmWave} integrated sensing and communications,'' in \emph{IEEE 25th Int. Workshop Signal Process. Adv. Wireless Commun. (SPAWC)}, Sep. 2024, pp. 951--955.

\bibitem{10530995}
M.~Nerini and B.~Clerckx, ``Physically consistent modeling of stacked intelligent metasurfaces implemented with beyond diagonal {RIS},'' \emph{IEEE Commun. Lett.}, vol.~28, no.~7, pp. 1693--1697, May 2024.

\bibitem{raeisi2024efficient}
M.~Raeisi, H.~Chen, H.~Wymeersch, and E.~Basar, ``Efficient localization with base station-integrated beyond diagonal ris,'' in \emph{IEEE Int. Conf. Commun. (ICC)}, 08-12 Jun. 2025.

\bibitem{10188340}
M.~Raeisi, A.~Koc, I.~Yildirim, E.~Basar, and T.~Le-Ngoc, ``Antenna array structures for enhanced cluster index modulation,'' in \emph{Joint Eur. Conf. Netw. Commun. \& 6G Summit (EuCNC/6G Summit), Gothenburg, Sweden}, 06 -- 09 Jun. 2023, pp. 102--107.

\bibitem{ming2025hybrid}
Z.~Ming, S.~Shen, J.~Rao, Z.~Li, J.~Zhang, C.~Y. Chiu, and R.~Murch, ``A hybrid transmitting and reflecting beyond diagonal reconfigurable intelligent surface with independent beam control and power splitting,'' \emph{arXiv preprint arXiv:2504.09618}, 2025.

\bibitem{10158690}
J.~An \emph{et~al.}, ``Stacked intelligent metasurfaces for efficient holographic {MIMO} communications in {6G},'' \emph{IEEE J. Sel. Areas Commun.}, vol.~41, no.~8, pp. 2380--2396, Jun. 2023.

\bibitem{10643881}
H.~Niu \emph{et~al.}, ``Stacked intelligent metasurfaces for integrated sensing and communications,'' \emph{IEEE Wireless Commun. Lett.}, vol.~13, no.~10, pp. 2807--2811, Aug. 2024.

\bibitem{10557708}
J.~An \emph{et~al.}, ``Two-dimensional direction-of-arrival estimation using stacked intelligent metasurfaces,'' \emph{IEEE J. Sel. Areas Commun.}, vol.~42, no.~10, pp. 2786--2802, Jun. 2024.

\bibitem{10622385}
H.~Liu, J.~An, D.~W.~K. Ng, G.~C. Alexandropoulos, and L.~Gan, ``{DRL-Based} orchestration of multi-user {MISO} systems with stacked intelligent metasurfaces,'' in \emph{IEEE Int. Conf. Commun. (ICC)}, 09-13 Jun. 2024, pp. 4991--4996.

\bibitem{akdeniz2014millimeter}
M.~R. Akdeniz, Y.~Liu, M.~K. Samimi, S.~Sun, S.~Rangan, T.~S. Rappaport, and E.~Erkip, ``Millimeter wave channel modeling and cellular capacity evaluation,'' \emph{IEEE J. Sel. Areas Commun.}, vol.~32, no.~6, pp. 1164--1179, Jun. 2014.

\bibitem{ying2020gmd}
K.~Ying, Z.~Gao, S.~Lyu, Y.~Wu, H.~Wang, and M.-S. Alouini, ``{GMD}-based hybrid beamforming for large reconfigurable intelligent surface assisted millimeter-wave massive {MIMO},'' \emph{IEEE Access}, vol.~8, pp. 19\,530--19\,539, 2020.

\bibitem{el2014spatially}
O.~El~Ayach, S.~Rajagopal, S.~Abu-Surra, Z.~Pi, and R.~W. Heath, ``Spatially sparse precoding in millimeter wave {MIMO} systems,'' \emph{IEEE Trans. Wireless Commun.}, vol.~13, no.~3, pp. 1499--1513, Jan. 2014.

\bibitem{9716880}
{\"O}.~T. Demir, E.~Bj{\"o}rnson, and L.~Sanguinetti, ``Channel modeling and channel estimation for holographic massive {MIMO} with planar arrays,'' \emph{IEEE Wireless Commun. Lett.}, vol.~11, no.~5, pp. 997--1001, Feb. 2022.

\bibitem{9693928}
M.~Cui and L.~Dai, ``Channel estimation for extremely large-scale {MIMO}: Far-field or near-field?'' \emph{IEEE Trans. Commun.}, vol.~70, no.~4, pp. 2663--2677, Jan. 2022.

\bibitem{10446555}
M.~A. Haider, Y.~D. Zhang, and E.~Aboutanios, ``Channel estimation and prediction in wireless communications assisted by semi-passive ris,'' in \emph{2024 IEEE Int. Conf. Acoust. Speech Signal Process. (ICASSP)}, Apr. 2024, pp. 8601--8605.

\bibitem{peng2022two}
C.~Peng, H.~Deng, H.~Xiao, Y.~Qian, W.~Zhang, and Y.~Zhang, ``Two-stage channel estimation for semi-passive ris-assisted millimeter wave systems,'' \emph{Sensors}, vol.~22, no.~15, p. 5908, Jul. 2022.

\bibitem{raeisi2023plug}
M.~Raeisi, I.~Yildirim, M.~C. Ilter, M.~Gerami, and E.~Basar, ``Plug-in {RIS}: A novel approach to fully passive reconfigurable intelligent surfaces,'' \emph{IEEE Trans. Wireless Commun.}, vol.~23, no.~10, pp. 14\,776--14\,789, Jul. 2024.

\bibitem{10453384}
M.~Nerini, S.~Shen, H.~Li, and B.~Clerckx, ``Beyond diagonal reconfigurable intelligent surfaces utilizing graph theory: Modeling, architecture design, and optimization,'' \emph{IEEE Trans. Wireless Commun.}, vol.~23, no.~8, pp. 9972--9985, Feb. 2024.

\bibitem{10159457}
H.~Li, S.~Shen, and B.~Clerckx, ``A dynamic grouping strategy for beyond diagonal reconfigurable intelligent surfaces with hybrid transmitting and reflecting mode,'' \emph{IEEE Trans. Veh. Technol.}, vol.~72, no.~12, pp. 16\,748--16\,753, Jun. 2023.

\bibitem{10187688}
I.~Santamaria, M.~Soleymani, E.~Jorswieck, and J.~Gutiérrez, ``{SNR} maximization in beyond diagonal {RIS}-assisted single and multiple antenna links,'' \emph{IEEE Signal Process. Lett.}, vol.~30, pp. 923--926, Jul. 2023.

\bibitem{9086460}
A.~Koc, A.~Masmoudi, and T.~Le-Ngoc, ``{3D} angular-based hybrid precoding and user grouping for uniform rectangular arrays in massive {MU-MIMO} systems,'' \emph{IEEE Access}, vol.~8, pp. 84\,689--84\,712, May 2020.

\bibitem{10694297}
{\"O}.~T. Demir and E.~Björnson, ``User-centric cell-free massive {MIMO} with {RIS}-integrated antenna arrays,'' in \emph{IEEE 25th Int. Workshop Signal Process. Adv. Wireless Commun. (SPAWC)}, Sep. 2024, pp. 546--550.

\bibitem{10622963}
J.~An \emph{et~al.}, ``Stacked intelligent metasurface performs a {2D} {DFT} in the wave domain for {DOA} estimation,'' in \emph{IEEE Int. Conf. Commun. (ICC)}, Jun. 2024, pp. 3445--3451.

\end{thebibliography}

\end{document}